\documentclass[a4paper,nofootinbib,preprintnumbers,twocolumn,preprintnumbers,floatfix,superscriptaddress,pra,showpacs]{revtex4-1}

\usepackage{graphicx}   
\usepackage{verbatim}  
\usepackage{color}      
\usepackage{subfigure}  
\usepackage[colorlinks=true,urlcolor=blue]{hyperref}   
\usepackage{float}
\usepackage{sidecap}
\usepackage{cleveref}
\usepackage[utf8]{inputenc}

\usepackage{mathrsfs}
\usepackage{amsmath}
\usepackage{amssymb}
\usepackage{esint}
\usepackage[english]{babel}
\usepackage{amsthm}
\usepackage[T1]{fontenc}
\usepackage{ae,aecompl}
\newtheorem{thm}{Theorem}[section]
\newtheorem{thmm}{Theorem}[]
\newtheorem{prop}{Proposition}[section]
\newtheorem{crl}{Corollary}[section]

\makeatletter
	\renewcommand{\fnum@figure}{FIG.~\thefigure \,\,(colour online)}
\makeatother

\newcommand{\bra}[1]{\langle #1|}					%bra
\newcommand{\ket}[1]{|#1\rangle}					%ket
\newcommand{\abs}[1]{\left| #1 \right|} % for absolute value
\newcommand{\avg}[1]{\langle #1 \rangle} % for average
\newcommand{\dibraket}[2]{\left< #1 \vphantom{#2} \right| \left. #2 \vphantom{#1} \right>} % for Dirac brackets

\newcommand{\johs}[1]{{\color{blue} #1}}

\newcommand{\qts}[1]{``#1''}

\begin{document}

\title{Elementary test for non-classicality based on measurements of position and momentum}

\author{Luca Fresta}
\affiliation{Dipartimento di Fisica, Universit\`a degli Studi di Milano, I-20133 Milano, Italy}
\author{Johannes Borregaard}
\affiliation{The Niels Bohr Institute, University of Copenhagen, Blegdamsvej 	17, DK-2100 Copenhagen \O, Denmark}
\affiliation{Department of Physics, Harvard University, Cambridge, Massachusetts 02138, USA}
%\author{Eran Kot}
%\affiliation{The Niels Bohr Institute, University of Copenhagen, Blegdamsvej 	17, DK-2100 Copenhagen \O, Denmark}
\author{Anders S. S\o rensen}
\affiliation{The Niels Bohr Institute, University of Copenhagen, Blegdamsvej 	17, DK-2100 Copenhagen \O, Denmark}

\date{04 August 2015}

\begin{abstract}

We generalise a non-classicality test
described by Kot \emph{et al.} [\href{http://link.aps.org/doi/10.1103/PhysRevLett.108.233601}{Phys. Rev. Lett. \textbf{108}, 233601 (2010)}], which can be used to rule out any classical description of a physical system.  The test is based on  measurements of quadrature operators and works by proving a contradiction with the classical description in terms of a probability distribution in phase space. As opposed to the previous work, we generalise the test  to include states without rotational symmetry in phase space. Furthermore, we compare the performance of the non-classicality test with classical tomography methods based on the inverse Radon transform, which can also be used to establish the quantum nature of a physical system. In particular, we consider a non-classicality test based on the so-called filtered back-projection formula. We show that the general non-classicality test is conceptually simpler, requires less assumptions on the system and is statistically more reliable than the tests based on the filtered back-projection formula. As a specific example, we derive the optimal test for 
a quadrature squeezed single photon state and show that the efficiency of the test does not change with the degree of squeezing.    

\end{abstract}
\maketitle

\section{Introduction}

Nearly a century after its foundation~\citep{Dirac26, vonNeumann}, Quantum Mechanics is a well established theory for microscopic and mesoscopic phenomena.
Nonetheless, the logical structure of quantum mechanics is still challenging our perception of the theory, spurring investigations of a deeper understanding of the quantum theory 
and its relation to the classical perception of reality~\citep{thooft14, Frohlich12}.
%Although not everybody is convinced that a deeper understanding of 
%the foundations of Quantum Mechanics will be grasped \cite{Preskill00},
In particular, the conceptual differences between the quantum
and classical world remain a fundamental topic of investigation,
in part due to the potential applications of quantum phenomena for quantum information technology~\cite{NielsenChuang}. 
To this end, experiments have sought to demonstrate effects which are truly {\it non-classical}, but the challenge in this case is how to prove that something really is non-classical. 
Clearly, the mere consistency with predictions from quantum theory is not sufficient to claim the observation of genuine quantum effects as this does not rule out a classical description of the same phenomena, although the classical explanation may seem different from the quantum mechanical. 
On the contrary, throughout this article, we shall use the definition that an experimental observation of a {\it genuine non-classical} effect should only assume concepts from classical physics and then demonstrate that  these assumptions lead to a contradiction with the observed experimental result. 
In other words, the goal of the non-classicality test is to convince a physicist trained in classical physics, but completely ignorant of quantum theory, that his/her perception of the world is incorrect.  A non-classicality test fulfilling this requirement  based on measurements of quadratures operators, $e.g.$, position $x$ and momentum $p$, was presented in Ref.~\cite{Kot12}. Furthermore, the test was used to demonstrate that the electric and magnetic field of a single photon state, cannot be described classically. In this article, we generalise the approach of Ref.~\cite{Kot12} and develop the theory for genuine non-classicality tests based on measurements of quadratures operators for a physical system. 
As opposed to Ref.~\cite{Kot12}, which only considered  rotationally symmetric states, we consider states of arbitrary shape. Furthermore, we compare the obtained results with more conventional approaches relying on tomographic reconstructions, $e.g.$ inverse Radon transformations, and demonstrate that the tests developed here are both conceptually simpler, require fewer assumptions on the system, and are statistically more reliable.

One of the most impressive examples of a genuine non-classical effect is the violation of Bell's inequality. In Bell's inequality, two simple classical assumptions are made, locality and realism, and this is shown to lead to a contradiction with the observed experimental results. Remarkably, this can be done without any assumptions about the physics of the underlying physical system. On the other hand, Bell's inequality is often quite difficult to violate in practice~\citep{Holevo11, Cavalcanti14}. We, therefore, take a different approach and consider non-classicality tests based on measurements of quadrature quantities of the form 
\begin{equation}
Q_{\theta}\equiv \cos\theta\, x + \sin\theta\, p
\label{eq: quadrature}
\end{equation}
for a single system. We shall assume that $Q_{\theta}$ exist for any $\theta \in \left[0,\pi \right]$, in the sense that they are well defined measurable observables of the physical system.
In particular, we are assuming that there exist two independent observables $x$ and $p$,
which might be canonical variables although this is not necessary. 
As considered in Ref.~\cite{Kot12}, this is for example the situation for optical fields measured by homodyne detection. 
In this case, a completely classical analysis reveals that $Q_{\theta}$ are the measured observables with the angle $\theta$ controlled by the phase of an interfering laser field and $x$ and $p$ are, $e.g.$, the electric and magnetic field components of a single mode of the measured optical field. 
Another example can be found in the emerging field of opto-mechanics~\citep{Optomechanics09,Optomechanics10} where a lot of effort is devoted to proving that mechanical oscillators can have non-classical behaviour. Here, a measurement of the position $x_m(t)$ at a time $t$ will have the form in Eq.~(\ref{eq: quadrature}) with $\theta=\omega_m t$ where $\omega_m$ is the vibration frequency of the oscillator and $x$ and $p$ are the (suitably rescaled) position and momentum operators at time $t=0$. Finally, operators of the form in Eq.~(\ref{eq: quadrature}) can also be measured for atomic ensembles where $x$ and $p$ represents two different spin components of the ensemble~\citep{Anders10}. 
In principle, operators of the form in Eq.~(\ref{eq: quadrature}) can be used to show a violation of Bell's inequality~\citep{Grangier03} but this requires preparing quantum states, which are highly difficult to produce. 
To avoid this, we make a stronger assumption compared to Bell's inequalities and assume that the measured observables are described by the relation in Eq.~(\ref{eq: quadrature}) as prescribed by the underlying physics of the system. This assumption makes our non-classicality inherently weaker than Bell's inequality, but on the other hand enables the violation of the non-classicality test under less stringent experimental requirements. 

As a second assumption underlying \emph{any} classical description of the system, we assume that the quadratures defined by Eq.~(\ref{eq: quadrature})
commute with one another. As a result, these two hypotheses  give the complete picture of the observable quantities that we will need for our non-classicality test.  
Mathematically speaking, it means that we postulate the algebra of observables
to be the commutative algebra generated by the rotated quadratures.
These hypotheses underlie the classical description of the system, 
which we will be able to rule out by deriving a contradiction from experimental data.
In this description where $x$ and $p$ are commuting quantities, the state
of the system will be characterized by a joint $xp$-probability distribution, also called the 
phase space distribution $W_c(x,p)$. Then, the contradiction consists in showing that such a distribution
does not exist for the measured experimental outcomes.

In quantum mechanics, it is well known that there does not exist a joint probability distribution for $x$ and $p$ since the position and momentum of a particle do not commute. 
The closest thing that one can have is the Wigner function $W(x,p)$~\citep{Glaubook}, which correctly predicts the probability outcome of measurements of the observables $Q_{\theta}$ but which cannot be considered a proper distribution function since it can attain negative values (see Fig.~\ref{fig: wigner plot}). 
The negativity of the Wigner function will be the underlying  quantum mechanical reason for the failure of the classical description in the non-classicality test derived below. 
In quantum mechanics, the Wigner function is, however, not the only possible phase space distribution and numerous others have been considered as the basis for non-classicality criteria. 
For instance the $P$-function represents the expansion of a field on coherent states and is thus negative for any state with an uncertainty squeezed below the level of a coherent state. Since, according to quantum theory, classical light sources produce coherent states, any state produced by classical sources will have a positive $P$-function. As a consequence, states with negative $P$-functions are often referred to as non-classical and based on this definition numerous observable non-classicality criteria have been derived~\citep{Vogel99,Vogel00,shapiro,Vogel04}. 
These criteria, however, inherently rely on quantum mechanics since the underlying criteria is based on having smaller uncertainty than expected from the quantum mechanical uncertainty relation.
 With the  definition used here, $i.e.$, that a non-classicality criteria should only rely on concepts from classical physics, tests based on the negativity of the $P$-function are thus not applicable. From a mathematical perspective, however, the test that we will derive is highly related. In particular our test is essentially the same as the one derived for the $P$-function in Refs.~\citep{Vogel99,Vogel04}, but it is translated to Wigner functions in such a way that quantum mechanics is not assumed, thus leading to a stronger test.

%PLOT OF THE WIGNER FUNCTION AND TOPVIEW WITH THE CUTS REPRESENTED
\begin{figure*}[t!]

\textbf{(a)}\includegraphics[scale=0.44]{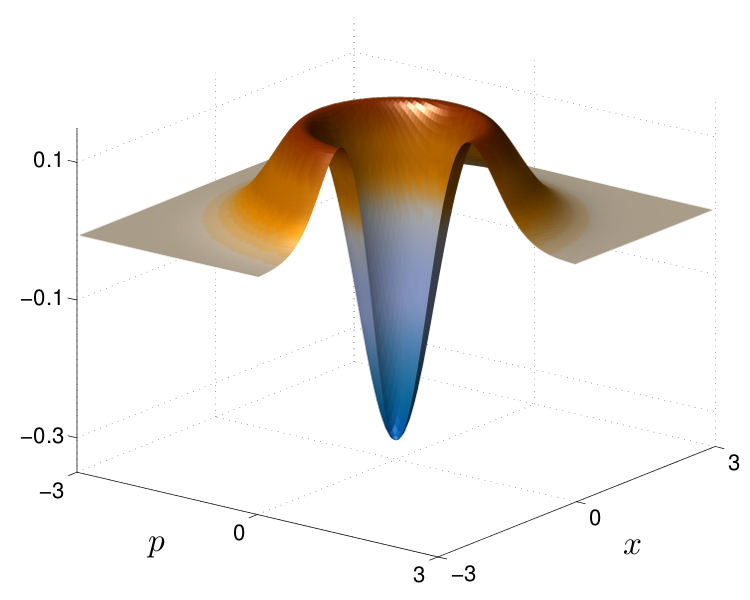}\textbf{(b)}\includegraphics[scale=0.64]{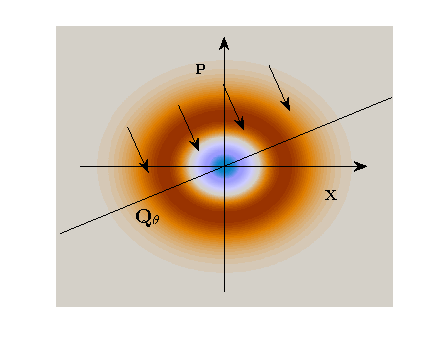}

\caption{ \textbf{(a)} Wigner function $W(x,p)$ of the single photon state of a single mode of light. It is not a probability distribution as it should be in the classical framework, since it takes on negative values. \textbf{(b)} Top view of the same function. The colour gradient allows to identify the negative region in the centre. The cut represented by the line is the quadrature $Q_{\theta}$. The marginal distribution corresponding to the probability distribution for measurement outcomes for  $Q_{\theta}$ is 
obtained by integration of $W(x,p)$ along the perpendicular directions to the cut, as shown by the arrows.}

\label{fig: wigner plot}
\end{figure*}

\section{Simple non-classicality test}

\subsection{Premise}
Before proceeding, we shortly review the non-classicality test presented in Ref.~\cite{Kot12}, which is based on previous work by Bednoraz and Belzig~\citep{Belzig11}. We will emphasise the
fundamental ideas behind the test and outline how
it can be generalized into the more general test introduced below. 

As already discussed in the introduction, a genuine non-classicality criterion
should directly disprove the classical
description without logically relying on quantum mechanics. 
Because the state of the system is classically characterized by a probability distribution $W_c(x,p)$ over the phase space of $x$ and $p$, the expectation value of any function $\mathscr{A}(x,p)$ is
\begin{equation}
\langle\mathscr{A}\rangle\equiv\int\int dx\,dp\, W_c(x,p)\mathscr{A}(x,p).
\label{eq: dual formula}
\end{equation}
Mathematically speaking, Eq.~(\ref{eq: dual formula}) is the integral representation of positive, normalized, linear functionals on the $C^\star-$algebra generated by the observables.
Accordingly, any non-negative quantity $\mathscr{F}(x,p)\geq0$ should have a non-negative expectation value,
\begin{equation}
\langle\mathscr{F}\rangle\equiv\int\int dx\,dp\, W_c(x,p)\mathscr{F}(x,p)\geq0,
\label{eq: violation clas}
\end{equation}
because the integrand is non-negative, $W_c(x,p)$ being a probability distribution. The violation of Eq.~(\ref{eq: violation clas}) disproves $W_c(x,p)$ as a probability distribution
and thus rules out any classical description of the system. 
We cannot, however,  directly determine $\langle\mathscr{F}\rangle$ since we cannot measure $x$ and $p$ simultaneously and we have to infer it from measurements of $Q_{\theta}$.

In Ref.~\cite{Kot12}, the test function $\mathscr{F}$ was chosen to be the square of a polynomial in the squared radius $r^2 = x^2+p^2$ since the underlying physical state was assumed to be rotationally symmetric. When the polynomial in $r$ was expanded, the various moment of $r$ could be expressed in moments of the rotated quadratures $Q_{\theta}$. Consequently, the mean value $\langle\mathscr{F}\rangle$ together with its experimental error
were directly computable from the experimental outcomes.
In Ref.~\cite{Kot12}, the developed test was applied to the experimental data for an approximate single photon state, attaining the violation of Eq.~(\ref{eq: violation clas}) by nineteen standard deviations~\cite{Kot12}.

\subsection{General formulation}
We will now generalise the test function considered in Ref.~\cite{Kot12} so that it does not rely
on rotational symmetry.         
Consequently, the generalized formulation can exploit the freedom in the choice of phases for the quadratures to construct an optimal test function for an arbitrary quantum state. 

Let us define the problem in more rigorous terms. 
To formulate a test, we need a suitable class of non-negative test functions $\mathscr{F}$
defined on the phase space and a method to determine 
$\langle\mathscr{F}\rangle$, together with its empirical error, directly from the measured quadrature moments.
The general class of smooth non-negative functions on the phase space is too large for our purposes and we therefore restrict the test to non-negative polynomials in the variables $x$ and $p$. Note, that such polynomials can
be peaked in a region around the negativity of $W(x,p)$ and since $W(x,p)$ typically vanishes faster than any power law outside its significant support, the divergent terms of the polynomials are irrelevant (see Fig.~\ref{fig: result graph}).
Consequently, we do not loose much generality by restricting to polynomials, and as we will show below, it allows us to make an explicit representation of the test function in terms of quadrature quantities.
 
Motivated by these considerations, we consider test functions of the form
\begin{equation}
\mathscr{F}(x,p)=\left(1+\sum_{i=1}^{N}\sum_{k=0}^{i}D_{i;k}x^{i-k}p^{k}\right)^{2}. 
\label{eq: general polynomial}
\end{equation}
Here we have, for simplicity, put the zeroth order coefficient equal
to one, since it is irrelevant to the effectiveness of the test as discussed below.
Furthermore, we choose to ensure the positivity by only considering  squared polynomials. This choice simplifies the test as discussed in Appendix~\ref{sec: A}, but does not exhaust all
the possible, non-negative polynomials. This is related to the $17^{th}$ Hilbert problem~\citep{17Hilbert, Vogel04}. 
However, since a polynomial of sufficiently high order can approximate an analytical function sharply peaked at any particular value where $W(x,p)$ is negative, we believe that considering this restricted set does not  degrade the efficiency of the test significantly. 

To have an experimentally applicable test, $\mathscr{F}(x,p)$ needs to be expressed in terms of powers of measurable quantities. For a fixed $i$, $Q_{\theta}^{i}$  directly gives a sum of $\left(i+1\right)$ moments of the form $x^{i-k}p^{k}$  with $k=0,...,i$. Hence, if one measures $m\geq i+1$ different angles $\theta_j$, one can always combine those linear equations by simple linear algebra to arrive at
\begin{equation}
x^{i-k}p^{k}=\sum_{j=1}^{m}T_{i;k,j}Q_{\theta_{j}}^{i},
\label{eq: relation MQ}
\end{equation}
where $T_{i;k,j}$ are real-valued coefficients which have to satisfy the linear constraint
\begin{equation}
\sum_{j=1}^{m}\left(\begin{array}{c}
i\\
k
\end{array}\right)\cos^{i}\theta_{j}\tan^{k}\theta_{j}\:\: T_{i;k',j}=\delta_{k,k'}\qquad\forall i
\label{eq: linear equation-1}.
\end{equation}
Note that  at least $(i+1)$ linearly independent quadratures are needed to express all $(i+1)$ moments of $x$ and $p$,  but if $m>i+1$, the constraint does not give a unique solution and there is a freedom  in the choice of the $T_{i;k,j}$ coefficients. 

Before proceeding, we stress that we have assumed commutativity of $x$ and $p$ in attaining Eq.~(\ref{eq: relation MQ}), $i.e.$, we assume a classical description of the quadratures in terms of $x$ and $p$.
From Eq.~(\ref{eq: relation MQ}), we can express the test
function in terms of the quadratures $Q_{\theta}^{i}$, which are directly measurable by assumption. 
Accordingly, we have explicitly shown that the test function $\mathscr{F}(x,p)$ is in the algebra of observables
and through Eq.~(\ref{eq: relation MQ}), we are able to compute its mean value. Finally, it is worth
noting that this strikingly simple linear relation is the very heart
of the procedure we consider. This is the trick  that allows us to avoid the
complicated machinery of the inverse Radon algorithms for the inverse problem of finding $W(x,p)$ from the measured quadratures $Q_\theta$ as we will discuss in Sec. \ref{sec:InvRadon}. 

In summary, once the polynomial order of the test function is fixed to
$2N$, the test is characterized by three sets of parameters:
\begin{itemize}
\item the set of free coefficients $D_{i;k}$ that define the test function
over the phase space,
\item the set of $m\geq 2N+1$ free phases $\left\{ \theta_{j}\right\} $ which identify
the quadratures we measure,
\item the set of coefficients $T_{i;k,j}$, which allow to express the $x^{i-k}p^k$ moments in terms of $Q_{\theta_j}$. These are not completely free to vary
but has to satisfy the linear constraint in  Eq. (\ref{eq: linear equation-1}).
\end{itemize}
We will formally refer to a test, $\mathbb{T}$, as the string of parameters
\begin{equation}
\mathbb{T}\equiv\left(N,\: D_{i;k},\:\theta_{j},\: T_{i;k,j}\right).
\label{eq: test list}
\end{equation}

We will now describe, in more details, how to express the expectation value of the test function and
its relative empirical error in terms of the measurable quadratures $Q^i_{\theta_j}$.
Expanding the square in the test function in Eq.~(\ref{eq: general polynomial}) and using Eq.~(\ref{eq: relation MQ}),
we get
\begin{eqnarray}
\mathscr{F}(x,p) & = & 1+\sum_{i=1}^{2N}\sum_{k=0}^{i}C_{i;k}x^{i-k}p^{k}\nonumber \\
 & = & 1+\sum_{j=1}^{m}\sum_{i=1}^{2N}Q_{\theta_{j}}^{i}\left(\sum_{k=0}^{i}C_{i;k}T_{i;k,j}\right)\nonumber \\
 & = & 1+\sum_{j=1}^{m}\mathscr{H}_{j}\left(Q_{\theta_{j}}\right).
 \label{eq: pol in Q}
\end{eqnarray}
For simplicity, we have here
introduced the coefficients $C_{i;k}\equiv2D_{i;k}+\sum_{a+a'=i}\:\sum_{b+b'=k}D_{a;b}D_{a';b'}$,
where $D_{i;k}$ vanishes for $i>N$.
Furthermore, we have introduced the set of $m$ polynomials, $\mathscr{H}_{j}\left(Q_{\theta_{j}}\right)$,
(one for each phase cut) defined as
\begin{equation}
\mathscr{H}_{j}\left(Q_{\theta_{j}}\right)\equiv\sum_{i=1}^{2N}Q_{\theta_{j}}^{i}\left(\sum_{k=0}^{i}C_{i;k}T_{i;k,j}\right).
\end{equation}
Eq.~(\ref{eq: pol in Q}) shows that according to the relation in Eq.~(\ref{eq: relation MQ}),
the test function can be expressed as a polynomial of quadratures, $i.e.$, as a sum of $2N+1$  polynomials of
degree $2N$, one for each independent phase cut. 

To evaluate the expectation value of the final expression in Eq.~(\ref{eq: pol in Q}), we consider each quadrature as an independent random variable distributed
according to its marginal. 
In principle, if one has a very large number of cuts (in particular if the number of cuts $m$ exceeds the minimal required value $m \geq 2N+1$), one could argue that nearby cuts will have similar moments and this may be used in a statistical analysis to reduce the variance on the experimental data. On the other hand, we prefer to make as few assumptions as possible and therefore consider each cut to be an independent  variable, distributed according to its marginal. 
Thus, for fixed phase cut, the expectation value $\langle\mathscr{H}_j(Q_{\theta})\rangle$
is formally given by integration with the marginal distribution and is experimentally estimated as
\begin{equation}
\langle\mathscr{H}_j(Q_{\theta})\rangle \simeq \frac{\sum_{n=1}^{M_j}\mathscr{H}\big( \big[ Q_{\theta}\big]_n \big)}{M_j},
\end{equation}
with $\big[ Q_{\theta}\big]_n$ being the $n$-th outcome within the dataset consisting of a total of $M_j$ measurements of the
random variable $Q_{\theta}$. 
Accordingly, we can compute the average value of the test function $\mathscr{F}$ as 
\begin{equation}
\langle \mathscr{F} \rangle = 1+\sum_{j=1}^{m} \langle \mathscr{H}_{j}\left(Q_{\theta_{j}}\right) \rangle,
\label{eq: mean value experimental}
\end{equation}
which according to a classical description, must be positive (See Eq.~(\ref{eq: violation clas})). 
%On the contrary, as shown in Refs.~\citep{Kot12,Belzig11} and discussed in detail below, $\langle \mathscr{F} \rangle$ can acquire negative %values, which proves that the system cannot be described within the framework of classical physics.

Quantum mechanically, $\langle \mathscr{F} \rangle$ is given by Eq.~(\ref{eq: dual formula}) with $W_c(x,p)$ being the Wigner function $W(x,p)$ describing the phase-space distribution. For a general density matrix, $\hat{\rho}$, the Wigner function is defined as
\begin{equation}
W(x,p)=\frac{1}{\hbar \pi}\int_{-\infty}^{\infty}\bra{x+y}\hat{\rho}\ket{x-y}e^{-2ipy/\hbar}dy,
\end{equation}  
where $\ket{x\pm y}$ are eigenstates of the position operator with eigenvalue $x\pm y$. The Wigner function is often referred to as a quasi-probability distribution since it can attain negative values and is thus not a true probability distribution. This possibility of the Wigner function being negative is the key aspect of our non-classicality test. With a negative Wigner function, Eq.~(\ref{eq: dual formula}) and the equivalent form in Eq.~(\ref{eq: mean value experimental}) are  no longer guaranteed to be positive leading to contradiction with classical physics. 
Given the Wigner function describing a quantum system, we can find the set of coefficients $D_{i;k}$ that minimizes $\langle \mathscr{F} \rangle$, by solving a linear optimization problem as described in Ref.~\cite{Kot12}. In an experimental test, however, we also have to consider the 
empirical error of our estimate of $\langle \mathscr{F} \rangle$. Since each random variable, $Q_{\theta}$,
is treated as statistically independent for different $\theta$, we can estimate the variance of $\langle \mathscr{F} \rangle$ as
\begin{equation}
\sigma_{\mathscr{F}} ^2 = \sum_{j=1}^{m} \frac{1}{M_c-1} \Big( \langle\mathscr{H}_j^{2}(Q_{\theta})\rangle 
   -  \langle\mathscr{H}_j(Q_{\theta})\rangle ^2 \Big),
\label{eq: err function}
\end{equation}
where we have assumed for simplicity that the same number of measurements $M_c$ is performed for each cut. This is sensible because of the statistical independence assumed, which means that no bias is justified (for a formal justification of this choice see Appendix~\ref{sec: A}).

We are interested in classifying how well the various non-classicality tests perform.
To do so, we introduce the following quantity for any non-classicality test,
\begin{equation}
\mathcal{G} \big[ \mathbb{T}\big] \equiv  \,\, \lim_{M\rightarrow\infty} \frac{\langle\mathscr{F}\rangle}{\sigma_{\mathscr{F}}\sqrt{M-m}},
\label{eq: violation ratio}
\end{equation}
which depends on all the characteristics of the particular test $\mathbb{T}$, but not on the total number of measurements $M \equiv m\,M_c$ because $\sigma_{\mathscr{F}}$ scales asymptotically as $1/\sqrt{M}$. In fact, as shown in Appendix~\ref{sec: A}, $\sigma_{\mathscr{F}}$ goes exactly as $1/ \sqrt{M-m}$  for the optimal test and the $M$ dependence thus disappears even before taking the limit. Note that because $\mathcal{G}$ involves a ratio, any overall proportionality constant in the test function
is irrelevant. We have therefore chosen to set the first coefficient in the polynomial in Eq.~(\ref{eq: general polynomial}) equal to unity to get rid of any overall scaling factor. 

The optimal test among all possible tests $\mathbb{T}$ (See Eq.~(\ref{eq: test list})) is defined as the
one, which minimizes $\mathcal{G}$ in Eq.~(\ref{eq: violation ratio}).
Formally, the search for the optimal test involves an algorithm to choose a suitable test among a class of logically equivalent tests.
Such an optimization is generally hard to do analytically and we will therefore treat it numerically.
Since we want to
prove the non-classicality of the system, it is natural to optimize the test using the
quantum mechanical description of its state, $i.e.$, by assuming the suitable Wigner function 
to be the correct description of the unknown distribution $W(x,p)$. 
We can then employ simple symmetry considerations about $W(x,p)$ in the optimization,
as we show below for rotational invariant states and squeezed states. Note, however, that even though the optimization assumes the quantum mechanical description of the state, the test itself is still completely classical, and the assumption of quantum mechanics at this stage does not decrease the strength of the test. In an experiment, the optimal choice of phase cuts 
has to be determined before the experimental measurements while the optimization 
of the other free parameters can be performed afterwards and may be optimized over the experimental data. If this is done, the statistical analysis of the experiments should also include the freedom in varying these parameters. For simplicity, we will just assume that one uses the coefficients expected from theory, in which case we do not have to worry about this complication. 
	
\subsection{Rotationally invariant states}

If the distribution $W(x,p)$, describing the state, is rotationally invariant, we naively expect that all the parameters of the optimal test function should exhibit this rotationally invariance, $e.g.$, that the optimal phase cuts $\{\theta_{j}\}$ should be  uniformly distributed. This was assumed in Ref. \cite{Kot12} but not formally shown. We will now show that this is indeed the case.  

We explicitly define the rotational version of the test function as
\begin{equation}
\mathscr{F}(x,p)=\left(1+\sum_{i=1}^{\left\lfloor N/2\right\rfloor }d_{i}r^{2i}\right)^{2},
\label{eq: radial functions}
\end{equation}
with $r^{2}=x^{2}+p^{2}$ being the squared radius in the phase plane.
Apart from the fact that we have formally reduced the set of coefficients
$D_{i;k}$ to the lighter set of $d_{i}$, we make a further reduction when  we relate the radial moments
to the quadrature moments in a similar way as done in Eq.~(\ref{eq: relation MQ}). We seek
a linear relation of the form
\begin{equation}
r^{2i}=\sum_{j=0}^{m}t_{i;j}Q_{\theta_{j}}^{2i},
\label{eq: radial functional relation}
\end{equation}
with the set of coefficients, $t_{i;j}$, being the analogue of the set of $T_{i;k,j}$ in Eq.~(\ref{eq: relation MQ}).
The existence of such a linear relation, provided a sufficient number of phase cuts, can be deduced from arguments similar to the ones leading to the relation in Eq.~(\ref{eq: relation MQ}).
Note that odd powers of the quadratures do not appear in the above
formula, since the parity symmetry of $r^{2i}$ guarantees
the null contribution of such terms. 
Denoting rotationally invariant tests as $\mathbb{T}_{rad}$,
we prove in Appendix \ref{sec: A} the following theorem.

\begin{thmm}
Whenever $W(x,p)$ is rotationally symmetric, and takes
on negative values, the optimal test, $i.e.$, the test $\mathbb{T}$
that minimizes $\mathcal{G}$, is a radial test $\mathbb{T}_{rad}$
and the cuts are chosen in such a way that
\begin{eqnarray}
\theta_{j} & = & \frac{\pi j}{m}\qquad j=0,...,m-1
\label{eq:uniform cuts}\\
t_{i;j} & \equiv t_{i} & =\left(\begin{array}{c}
2i\\
i
\end{array}\right)^{-1}\frac{4^{i}}{m}\qquad\forall j.
\label{eq:uniform coefficients}
\end{eqnarray}
Furthermore $\mathcal{G}$ is independent of the number of cuts provided that $m\geq N+1$.
\label{thm: 1}
\end{thmm}

The theorem states that the optimal test function is a radial polynomial and the cuts
should be distributed uniformly over the phase space whenever $W(x,p)$ is rotationally symmetric. 
Note that the numerical optimization of the test is a lot simpler than in the general
case because the optimal test function can be put into the form in Eq.~(\ref{eq: radial functions}),
which has fewer free parameters.
Yet, the reader should bear in mind that with the choice given by Eqs.~(\ref{eq:uniform cuts}) and (\ref{eq:uniform coefficients}), the linear relation in Eq.~(\ref{eq: radial functional relation}) holds
and can be applied even for non-symmetric states. Hence, the assumptions made can never lead to a state incorrectly being assigned as non-classical even if the assumptions are not fulfilled. Applying a rotationally invariant test to a non rotationally invariant state will thus only lead to a less efficient test but not to an incorrect result.

For the test $\mathbb{T}_{rad}$, we can express the test function, using Eqs.~(\ref{eq: radial functional relation}) and (\ref{eq:uniform coefficients}), as
\begin{eqnarray}
\mathscr{F} & =1+ & \sum_{j=0}^{m}\mathscr{H}\left(Q_{\theta_{j}}\right),
\end{eqnarray}
where $\mathscr{H}$ is now the same polynomial for all the phases, contrary to the general case in Eq.~(\ref{eq: mean value experimental}). Defining $c_{i}=2d_{i}+\sum_{a+a'=i}d_{a}d_{a'}$, we can explicitly
express this polynomial as
\begin{equation}
\mathscr{H}\left(Q_{\theta}\right)\equiv\sum_{i=1}^{2N}c_{i}t_{i}Q_{\theta}^{2i}.
\end{equation}
Otherwise all considerations about the evaluation of $\mathcal{G}$
are the same as for the general test. Most importantly, we treat every cut as statistically independent, even in the case of actual rotational symmetry,
because we do not want to assume this symmetry.

For the single photon state considered in Ref.~\cite{Kot12}, which   corresponds to the first excited state of a harmonic oscillator, the Wigner function is
\begin{equation}
W(x,p) = \frac{1}{\pi} \Bigl( 2(x^2+p^2)-1\Bigr) \exp{[-(x^2+p^2)]}.
\end{equation}
We have numerically optimized $\mathcal{G}$ with respect to $\mathbb{T}$ for increasing polynomial order. 
Fig.~\ref{fig: result graph} shows the optimal $\mathcal{G}$ as a function 
of $N$, the order of the test function being $2N$.
We find a violation of Eq.~(\ref{eq: violation clas}) for $N\geq4$ in agreement with what was found in Ref.~\cite{Kot12}. 
Furthermore, we see that the graph plotted in Fig.~\ref{fig: result graph} exhibits an asymptotic saturation of $\mathcal{G}$. Accordingly, in the case of arbitrarily high number of measurements, increasing the polynomial degree beyond $N\sim16$ gives more degrees of freedom for the test but does not decrease $\mathcal{G}$ significantly. This demonstrates that polynomial test functions are efficient.  In principle, more general, bounded and smooth test functions could be employed but the results obtained here suggest that this is not necessary.

\begin{figure}[t!]
\includegraphics[scale = 0.43]{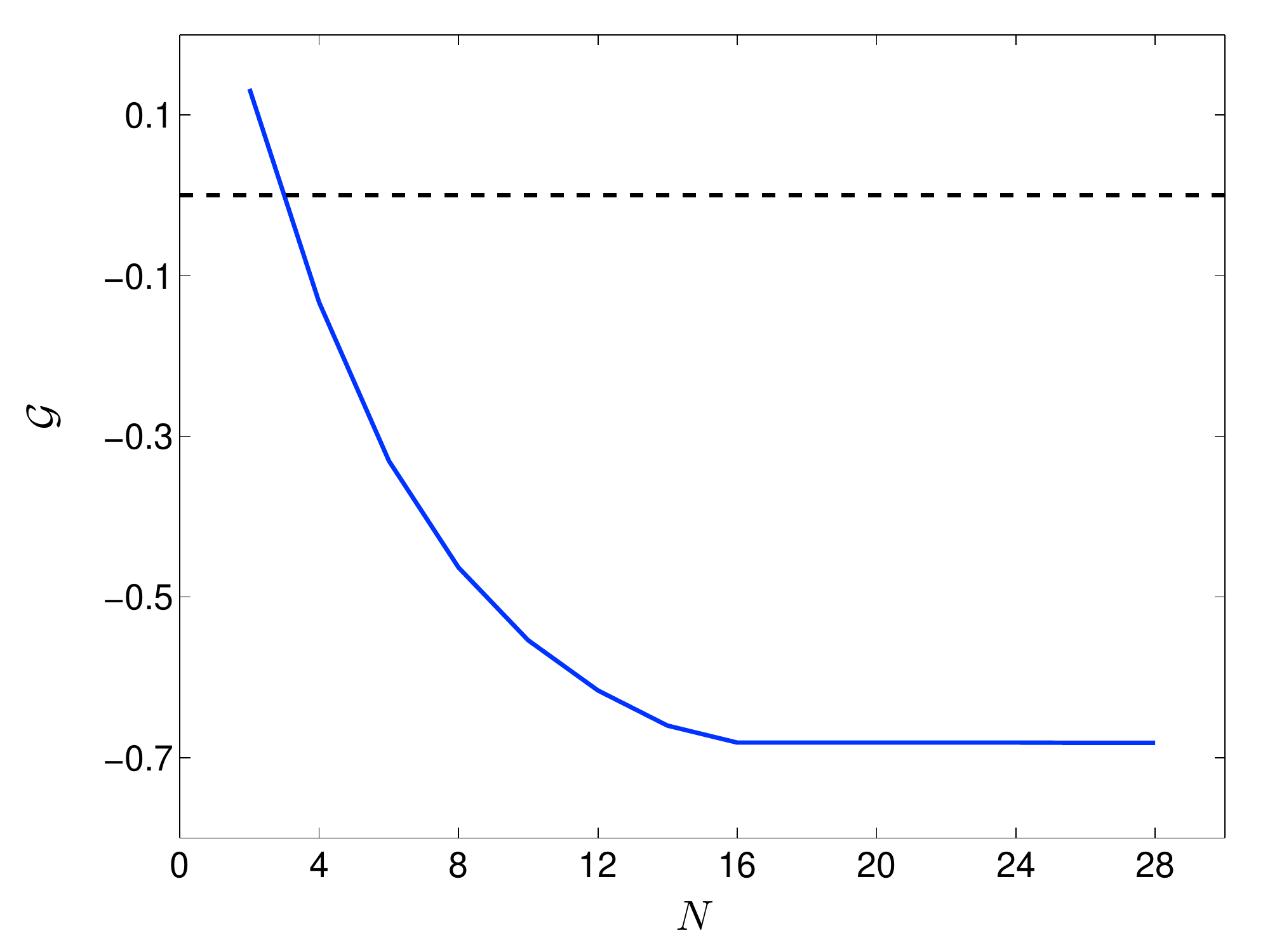}
\caption{Violation of non-classicality for the first excited state of a harmonic oscillator, e.g, a single photon state. The figure show the optimal expectation value of the test function relative to its standard deviation ($\mathcal{G}$) as a function of the order of the polynomial  in Eq. (\ref{eq: radial functions}). 
Experimentally this may result in a violation of classicality by $-\sqrt{M-m}\mathcal{G}$ standard deviations, with $M$ being the total number of measurements and $m\geq N+1$ the number of measured quadratures.
\label{fig: result graph}}
\end{figure}

While $\mathcal{G}$ quantifies the relation between different tests, it does not directly give the statistical strength, ${\langle\mathscr{F}\rangle}/{\sigma_{\mathscr{F}}}$ of a test. The latter can, however, be obtained by simply multiplying $\mathcal{G}$ with $\sqrt{M-m}$. As stated above $\mathcal{G}$ is independent of the number of cuts, but because of the dependence of $m$ in  $\sqrt{M-m}$ it will in general be better to have the least number of cuts, $i.e.$, $m=N+1$, although this dependency will be small if a large number of measurement are performed $M\gg m$. In order to keep $m$ as small as possible we thus find from Fig.~\ref{fig: result graph} that ${\langle\mathscr{F}\rangle}/{\sigma_{\mathscr{F}}}$ is minimised by working  around $N =16$ although higher order tests $N>16$ will produce essentially similar results if sufficiently many measurements are performed $M\gg m$.  For imperfectly prepared states, however, the area in phase space where the Wigner function is negative is smaller. Hence in this case it may be desirable to work at a higher $N$ where the function can be tailored to be more narrow \cite{Kot12}.
In Fig.~\ref{fig: optimal test function}, we plot the optimal test function for the minimum argument $N = 16$. It is seen that the test function is peaked at the negative region of the Wigner function and that the tails of the polynomial are outside the significant support of $W(x,p)$ and are thus not affecting the test.

\begin{figure}[t!]
\includegraphics[scale=0.30]{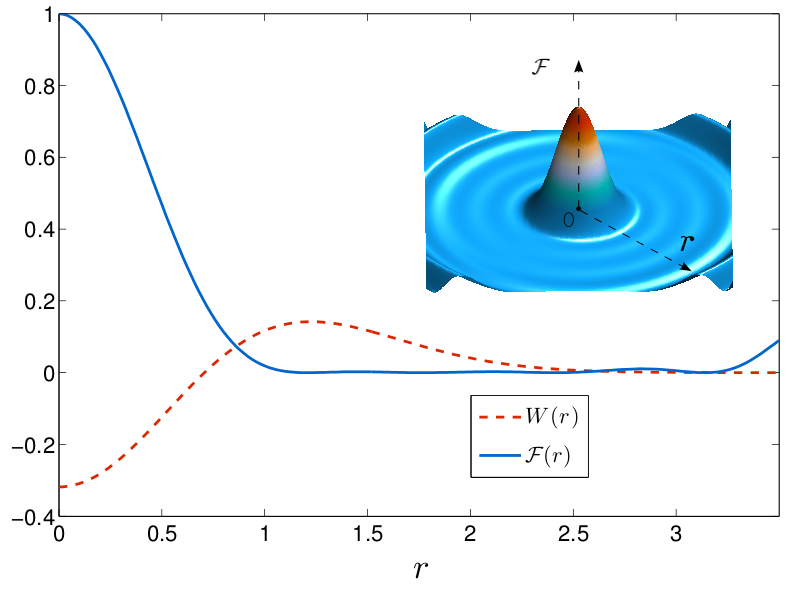}
\caption{Radial profile of $W(x,p)$ for the single photon
state (red dashed curve) and of the optimal test function $\mathscr{F}(x,p)$ (light blue solid curve) for $N=16$. The inset shows the corresponding 3D plot of $\mathscr{F}(x,p)$.
\label{fig: optimal test function}}
\end{figure}

\subsection{Squeezed states}

We now turn to the case where $W(x,p)$ is not rotationally symmetric. 
As a specific example, we consider 
a squeezed single photon state and show that the optimal test function differs 
significantly from a rotationally invariant test. 
Let $W(x,p)$ be the Wigner function of the single photon state. The Wigner function, $W_{\lambda}(x,p)$ of the squeezed state
is given by the rescaling
\begin{equation}
W_{\lambda}(x,p)=W(\lambda x,p/\lambda),
\label{eq: rescaling}
\end{equation}
where $\lambda$ characterizes the amount of squeezing in either $x$  ($\lambda>1$) or $p$ ($0<\lambda<1$).
The unsqueezed state corresponds to $\lambda=1$.
Fig. \ref{fig: squeezing of cuts} shows the result of the squeezing transformation; the Wigner function of the single photon state is stretched into an ellipse for the squeezed state.

The squeezed state still exhibit  elliptical symmetry but the marginal distributions are no longer identical as they were for the rotationally invariant state. The extra freedom in the choice of cuts present in the general test is therefore of importance.  
As an ansatz for a test function, we consider
\begin{equation}
\mathscr{F}_{\lambda}(x,p)=\left(1+\sum_{i=1}^{\left\lfloor N/2\right\rfloor }d_{i}r_{\lambda}^{2i}\right)^{2},
\label{eq: rescaled radial functions} 
\end{equation}
with $r_{\lambda}^{2}=\left(\lambda x\right)^{2}+\left(p/\lambda\right)^{2}$
being the rescaled radius in the phase plane. Eq.~(\ref{eq: rescaled radial functions}) is a rescaling of the radial functions in  Eq.~(\ref{eq: radial functions}) by the squeezing parameter, $\lambda$ . 
Note that with this type of test function, the average value
of $\mathscr{F}_{\lambda}$ with a phase space distribution $W_{\lambda}$ is independent of $\lambda$.
The polynomial expression of $\mathscr{F}_{\lambda}(x,p)$ with respect to the rotated quadratures is obtained in a straightforward manner from Eq.~(\ref{eq: radial functional relation}) with
the substitution $x\rightarrow\lambda x$ and $p\rightarrow p/\lambda$ such that		
\begin{eqnarray}
r_{\lambda}^{2i} & = & \sum_{j=0}^{m}t_{i;j}\left(\cos\theta_{j}\,\lambda x+\sin\theta_{j}\, p/\lambda\right)^{2i}\nonumber \\
 & = & \sum_{j=0}^{m}t_{i;j}\left(\lambda\cos\theta_{j}\right)^{2i}\left(x+\frac{\tan\theta_{j}}{\lambda^{2}}p\right)^{2i}\nonumber \\
 & = & \sum_{j=0}^{m}t_{i;j}\left(\lambda\frac{\cos\theta_{j}}{\cos\tilde{\theta}_{j}}\right)^{2i}Q_{\tilde{\theta}_{j}}^{2i}
\end{eqnarray}
where $\tilde{\theta}\in[0,\pi)$ satisfies
$\tan\tilde{\theta}_{j}=\frac{\tan\theta_{j}}{\lambda^{2}}$. Geometrically,
this transformation $\theta\rightarrow\tilde{\theta}$ is the polar
representation of the rescaling that we have performed on the phase
plane.

\begin{figure}[t!]
\flushleft
\includegraphics[scale=0.23	]{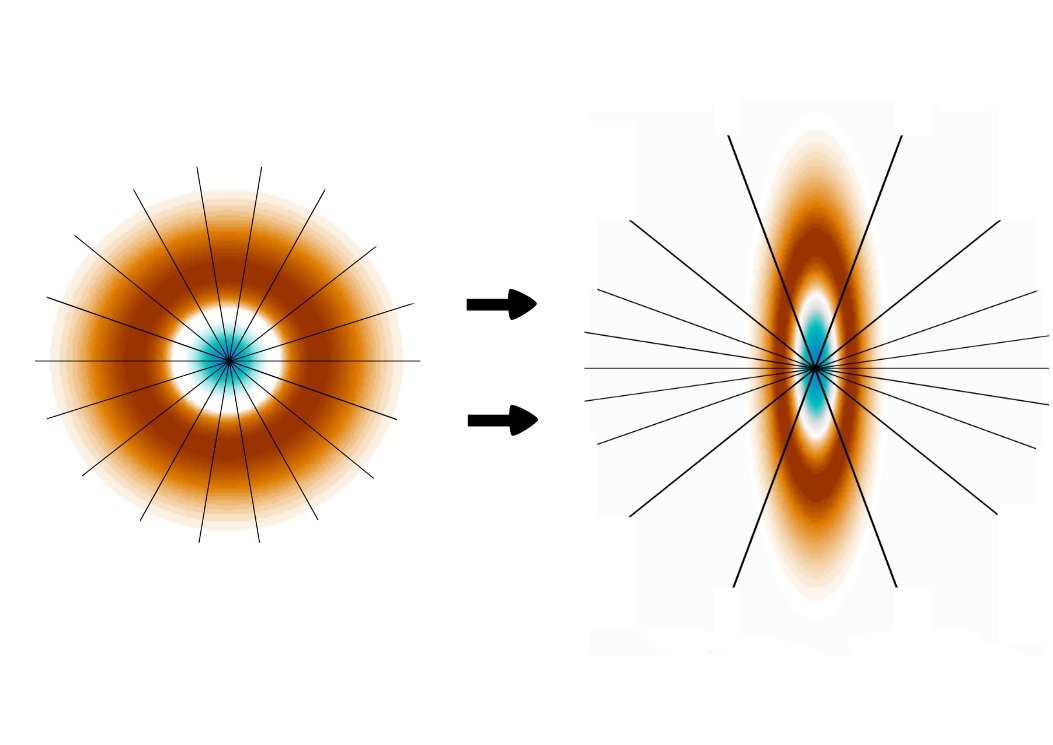}
\caption{Sketch (not to scale) of $W(x,p)$ (left) and $W_\lambda(x,p)$ (right), together with a representation of the optimal distribution of phase cuts. The cuts are uniformly distributed for the
single photon state and stretched along the $x$ axis for the squeezed state. Note that the cuts
stretch along the opposite direction of stretching for $W_\lambda(x,p)$. This reflects that most cuts are placed at angles where the distribution varies the fastest.  
%\johs{What value was used for $\lambda$?}\lu{I did it manually. FOrmally it is like I rescaled the final figure. This is so since otherwise to have a visible squeezing the major axis would be too high. Maybe it is better to say that left and right are on 2 different scales}
The inner/blue region is where the distribution takes on negative values.} 
\label{fig: squeezing of cuts}
\end{figure}

In Appendix \ref{sec: A}, we show that not only is the average value of $\mathscr{F}_{\lambda}$ independent of $\lambda$, but also the same is true for the figure of merit $\mathcal{G}$.
This is formally stated in the following theorem.
\begin{thmm}(Invariance under squeezing)
The minimum of $\mathcal{G}$ for the squeezed state is independent
of the squeezing parameter $\lambda$ and equal to the minimum for the rotationally symmetric state.
The optimal test function is given by Eq.~(\ref{eq: rescaled radial functions}) with $d_{i}$ being
the same as the optimal ones for the single photon state. Furthermore $\theta_{j}$ and $t_{i;j}$ 
given in Eqs.~(\ref{eq:uniform cuts},\ref{eq:uniform coefficients}) transform to the \qts{rescaled} values
\begin{eqnarray*}
\theta_{j} & \rightarrow & \tilde{\theta}_{j}\\
t_{i;j} & \rightarrow & \tilde{t}_{i;j}\equiv t_{i;j}\left(\lambda\frac{\cos\theta_{j}}{\cos\tilde{\theta}_{j}}\right)^{2i}
\end{eqnarray*}
\label{thm: 2}
\end{thmm}
In Appendix \ref{sec: A}, the theorem is stated in a more general fashion along with the proof of it. In other words, allowing the distribution of phase cuts to deviate from the uniform distribution
results in non-classicality tests, which are equally efficient for the squeezed and unsqueezed single photon states. The introduction of the squeezing thus does not degrade the test and the results of Fig. \ref{fig: result graph} are also applicable to the squeezed single photon state. This is in contrast with the reconstruction method based on standard inverse Radon transformation, where squeezing would introduce additional ambiguities and/or a reduction in the statistical significance of the results as discussed below.

\section{Inverse Radon Transform}
\label{sec:InvRadon}
\subsection{Introduction}

Above, we have presented a simple and general non-classicality test that
detects negativity of a phase space distribution $W(x,p)$ resulting in the breakdown of any classical description of the system.
This objective could, however, also be achieved through standard state reconstruction methods based on the inverse Radon transform.
As already stated in the previous section, we cannot directly measure the distribution $W(x,p)$ with arbitrary precision.  What we can measure is the probability distribution of the quadrature operators $Q_\theta$, which is given by the marginal distribution known as the Radon transform,
\begin{equation}
\mathbf{R}W(\theta,s) \equiv \int\int_{\mathbb{R}^2} W(x,p)\, \delta(s- x\cos\theta - p\sin\theta)dxdp.
\end{equation}
%I DELETED SOME STUFF HERE CHECK IF WE NEED IT
%Nonetheless, we can investigate its Radon transform, $\mathbf{R}W(\theta,s)$, with arbitrary precision. $\mathbf{R}W(\theta,s)$
%gives the marginal distribution for a measurement of the quadrature $Q_\theta$.
Formally, $\mathbf{R}W(\theta,s)$ %is the probability density for $Q_\theta=s$. It 
is an integral transform
obtained from integration along the rays perpendicular to the cut characterized by $\theta$ at 
distance $s$ from the origin. Fig.~\ref{fig: wigner plot}(b) shows graphically how this transform works.
The inverse process is referred to as the inverse Radon transformation, which allows to extract the properties of
the underlying distribution $W(x,p)$ from $\mathbf{R}W(\theta,s)$. 
We will now  consider the use of the inverse Radon transform to show 
non-classicality. Furthermore, we will compare it to the simple test discussed above in order to determine advantages and disadvantages of the two methods for an experimental demonstration of non-classicality.

\begin{figure}
\flushleft
\includegraphics[scale=0.18]{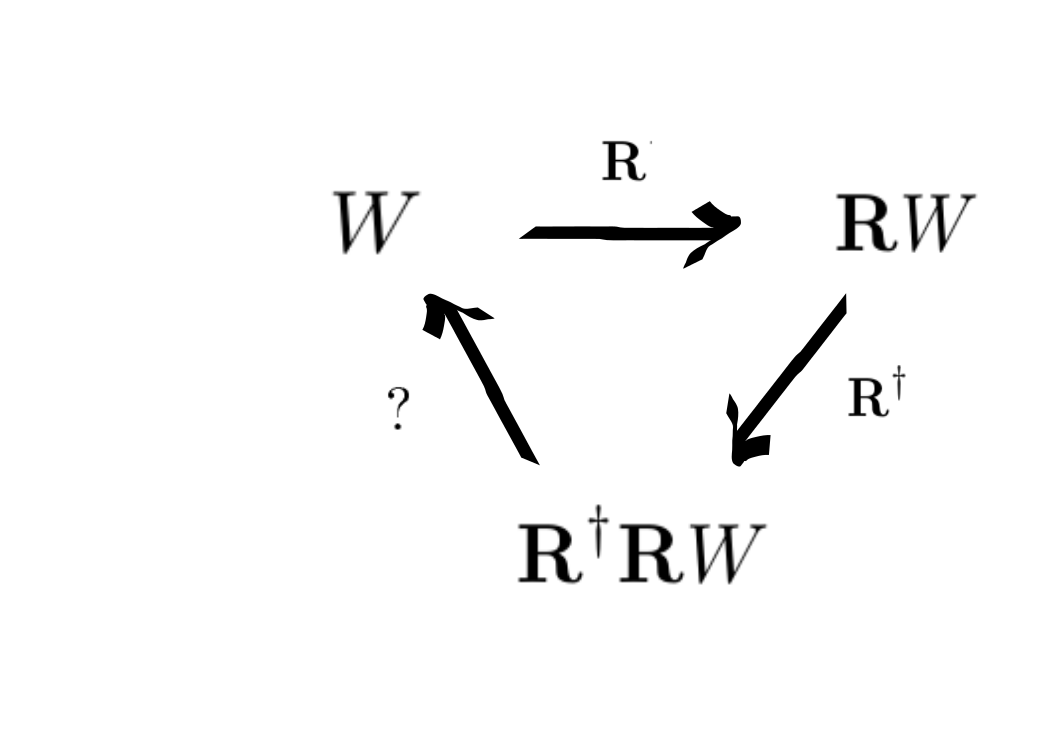} %trimm= 0cm 8cm 0cm 3cm -sorry got an error compiling
\caption{Action of the operators $\mathbf{R}$ and $\mathbf{R}^{\dagger}$. The inversion cannot be achieved directly through the application of these two operators, $i.e.$, there is a final step missing represented by the question mark. Since the latter is hard to achieve, the idea of the filtered back-projection is to rely on $\mathbf{R}^{\dagger}$ for making statements about $W(x,p)$.}
\label{fig:rrdagger}
\end{figure}

The inversion of the Radon transform is not a simple task, but because of its importance
in image reconstruction, there is a wide range of techniques to accomplish it.
For our purpose, the reconstruction of the complete distribution
$W(x,p)$ is not necessary, since we are only interested in the most efficient way to establish negativity of $W(x,p)$. This negativity of $W(x,p)$ only occurs in limited regions of the phase plane and we thus do not need the full information.
This motivates resorting to  the so-called filtered back-projection theorem \cite{Deans, Markoe}, which provides a technique
to evaluate the convolution of $W(x,p)$ with another function on the phase space. 
Here, we will only present the main ideas of the theorem and refer to Ref.~\cite{Deans} for details.
The Radon transform identifies a linear operator $\mathbf{R}$ on the functional space of
well-behaved functions defined on the phase space, to the space of
functions defined on the so-called unit cylinder, spanned by $(\theta,s)$.
For instance, $W(x,p)$ belongs to the first functional space while $\mathbf{R}W(\theta,s)$ belongs to the latter.
Both spaces can be equipped with suitable inner products and accordingly,
if the Radon transform operator is defined on functions like $W(x,p)$,
the definition of the adjoint operator
$\mathbf{R}^{\dagger}$ follows. It can be proven \cite{Deans, Markoe} that the action of $\mathbf{R}^{\dagger}$ on functions defined on the unit cylinder is given by
\begin{equation}
\mathbf{R}^{\dagger}\omega\;(x,p)=\frac{1}{\pi}\int_{-\frac{\pi}{2}}^{+\frac{\pi}{2}}\omega(\theta,x\cos\theta+p\sin\theta)d\theta.
\label{eq: adjoint radon}
\end{equation}
Despite this formula makes the adjoint Radon transformation easy to compute,
%While this formula defines the adjoint of the Radon transformation, 
it does not solve the problem of determining the distribution $W(x,p)$ because the adjoint $\mathbf{R}^\dagger$ is not the inverse of $\mathbf{R}$ (see Fig.~\ref{fig:rrdagger}). Instead of inversion, we can, however, use the adjoint to determine properties of $ W(x,p)$ from the filtered back-projection formula, which states that
\begin{multline}
\int_{\mathbb{R}^{2}}F(x_{0}-x,p_{0}-p)W(x,p)\textrm{d}x\textrm{d}p=\\
=\frac{1}{\pi}\int_{-\frac{\pi}{2}}^{+\frac{\pi}{2}}\int_{\mathbb{R}}\mathbf{R}W\;(\theta,x_{0}\cos\theta-p_{0}\sin\theta-s)\,\omega(\theta,s)\, dsd\theta,
\label{eq: filtered backprojection}
\end{multline}
provided that the filter, $F(x,p)$, and its kernel, $\omega(\theta,s)$,
are related by the adjoint Radon operator, that is 
\begin{equation}
\mathbf{R}^{\dagger}\omega\;(x,p)=F(x,p).
\label{eq: Cho's problem}
\end{equation}

The advantage of the filtered back-projection formula in Eq.~(\ref{eq: filtered backprojection})
in view of applications is straightforward. By integrating the right hand side, where 
$\mathbf{R}W\;(\theta,s)$ are the measured distributions and $\omega(\theta,s)$ is a known function, one
attains the convolution of the unknown $W(x,p)$ and the known filter
$F(x,p)$. For state reconstruction,
the idea is to avoid a direct cumbersome inversion of $\mathbf{R}$
by using Eq.~(\ref{eq: filtered backprojection}) to get a coarse graining through an \emph{ad hoc} designed filter.
It is, however, important to realize that the filtered back-projection theorem shifts
the problem of the inversion of the operator $\mathbf{R}$ to the problem of finding the kernel satisfying Eq.~(\ref{eq: Cho's problem}) for a given filter $F(x,p)$. This problem, also called Cho's problem in the literature, is highly non-trivial. For instance, a solution of Cho's problem
for $F(x,p)=\mathscr{F}(x,p)$ with $\mathscr{F}$ given in Eq.~(\ref{fig: optimal test function}), is not known in general. 

For our purpose, the filter $F(x,p)$, will play the role of a test function and Eq.~(\ref{eq: filtered backprojection}) gives us a recipe to compute its mean value, possibly translated in the plane, with respect to the unknown distribution $W(x,p)$. Hence, similar to the simple test above, the non-existence of a positive probability distribution $W(x,p)$ can be proven by the right hand side yielding a negative value for a non-negative filter $F(x,p)$.
The difficulty of solving Cho's problem, however, limits the choice of test functions that we can consider.

\subsection{Application}
In the previous non-classicality test, we looked for a violation of the condition in Eq.~(\ref{eq: violation clas}) 
for the squeezed and unsqueezed single photon state. Following a similar path of reasoning, we argue that a state is non-classical if it is possible to get a negative outcome by performing the convolution of $W$ with a non-negative filter function $F$ (see Eq.~(\ref{eq: filtered backprojection})) through the back-projection formula with a kernel $\omega$ as the solution to Cho's problem.  In the following, we will refer to the filter $F$ as the test function. 
A choice of filter is the normalized characteristic function of a disc, for which
Eq.~(\ref{eq: Cho's problem}) has been solved in the literature~\citep{Niev86},
\begin{equation}
F_{a}(x,p)=\begin{cases}
\frac{1}{\pi a^{2}} & \quad(x^{2}+p^{2})\leq a^{2}\\
0 & \quad(x^{2}+p^{2})>a^{2}
\end{cases}
\label{eq: filter approximate dirac measure}
\end{equation}
This filter is indeed non-negative and has a compact support so that it can
only have a negative expectation value if $W(x,p)$ takes
on negative values. This filter is not usually employed in tomography 
because it is not robust against noise. 
Other filters employed in reconstructions are
more robust against noise, but they are not non-negative. The convolution with a negative filter would introduce an ambiguity in the interpretation of a possible negative outcome, since one would have to determine if it came from the distribution $W$ or the filter $F$. We, therefore, choose to consider a completely positive filter of the type in Eq.~(\ref{eq: filter approximate dirac measure}) despite its inherent instability. 
Ideally, it would be better to employ a more smooth non-negative filter, but to our knowledge no simple solutions to Cho's problem has been determined for such filters. Using a circular disc is, however, only a good test function for the unsqueezed single photon state, which is also rotational invariant. The geometry of the squeezed state (see Fig.~\ref{fig: squeezing of cuts}) suggests that it is better to use a stretched filter 
$F_{a,l}(x,p)=F_{a}(l x,p/l)$ for this state, provided that we can determine
the corresponding kernel. Such  a stretched filter will have a  better overlap with the negative areas in the Wigner function of the squeezed state and we thus expect it to produce statistically more significant results.
Note that for the sake of generality, we here allow for  a stretching parameter $l$,
which might be different from the squeezing parameter $\lambda$.

The solution to Cho's problem for $F_{a}(x,p)$ given by Eq.~(\ref{eq: filter approximate dirac measure})
can be found in Ref.~\cite{Niev86} and reads,
\begin{eqnarray}
\omega_{a}(\theta,s)\equiv w_{a}(s)=\begin{cases}
\frac{1}{\pi a^{2}} & |s|\leq a\\
\frac{1}{\pi a^{2}}\,\left(1-\frac{1}{\sqrt{1-(a/s)^{2}}}\right) & |s|>a \qquad
\end{cases}
\label{eq: dilated approximate cylindrical dirac measure}
\end{eqnarray}
In Appendix \ref{sec: B}, we derive a method to determine the kernel
for any linear transformation of the filter. In particular for the stretched filter $F_{a,l}(x,p)$ we find the corresponding kernel
\begin{equation}
\omega_{a,l}(\theta,s)=\frac{1}{u^{2}(l,\theta)}\,\omega_{a}\left(\frac{s}{u(l,\theta)}\right),
\label{eq: general kernel}
\end{equation}
where we have defined the function 
\begin{equation}
u(\theta,l)=\frac{\cos\theta}{l}\sqrt{1+l^{4}\tan^{2}\theta}.
\end{equation}

In order to use the back-projection formula, we need to determine the integral on the rhs
of Eq.~(\ref{eq: filtered backprojection}).
 Since we know that the Wigner functions in consideration are negative around the origin, we set  $(x_{0},p_{0})=(0,0)$
and put $l=\lambda$. The motivation for the latter is that if the filter covers the largest possible area in phase space,
the error associated to the filtered back-projection is reduced, at least for rather small areas, since we are coarse graining over a larger area of the phase space.
On the contrary, a filter limited to a small region around the minimum of $W(x,p)$ will give a more negative result and thus a larger violation of classicality. The optimal filter is thus obtained as a trade off, which suggests a filter matching the symmetry of $W(x,p)$. Consequently, we  put $l=\lambda$.

We define $\langle F_{a,\lambda} \rangle$ as being the lhs of Eq.~(\ref{eq: filtered backprojection}), $i.e.$,
\begin{equation}
\langle F_{a,\lambda} \rangle = \frac{1}{\pi}\int_{-\frac{\pi}{2}}^{+\frac{\pi}{2}}\int_{\mathbb{R}} \mathbf{R}W_\lambda\;(\theta,s)\,\omega_{a,\lambda}(\theta,s)\, dsd\theta.
\label{eq: final back projection}
\end{equation}
Similar to Eq.~(\ref{eq: pol in Q}), this equation expresses the mean value of the positive function $F$ over the phase space in terms of measurable quantities since, for fixed $\theta$, the function $\mathbf{R}W_\lambda\;(\theta,s)$ is just the measured distribution of the quadrature $Q_\theta$. Nonetheless, the calculation of the mean value of the test function in Eq.~(\ref{eq: final back projection})
involves an integration over both the values of the quadrature and the phases. For the elementary test, only a finite set of at least $2N+1$ phase cuts was required and a negativity was already seen for $N=4$ in the single photon state. 
Here, on the other hand, evaluation of the filtered back-projection over a certain finite number of angles will necessarily involve some ambiguities since it is $a$ $priori$ not clear how many angles are required to get a certain precision. Since we do not want to assume $a$ $priori$ information about the phase space distribution, a proper analysis ought to include a study of the largest possible error introduced by restricting the number of angles.

As opposed to estimating the error introduced by the finite number of cuts, we will in the following section focus on an alternative strategy. We will refer to this as a Monte Carlo numerical quadrature, where we also consider the measurement angle as a random variable. In this way, the number of cuts is equal to the number of measurements and we thereby gradually remove the ambiguities  from the  finite number of cuts as the number of measurements increases. 
We find that this strategy is conceptually simpler to understand, but an alternative and more natural strategy of using a fixed number of phase space cuts might be desirable  experimentally. In the main text, we will only discuss this latter strategy briefly and leave most of the discussion to  Appendix~\ref{sec: B} where we  show that it gives similar results although at the cost of introducing unpleasant ambiguities.

For an experimental test, one is concerned with the statistical error in the evaluation of Eq.~(\ref{eq: final back projection}). This depends on the specific strategy adopted but there is a general issue in all the test we consider, which is related to our choice of kernel $\omega$.	
The discontinuity of the kernel (see Eq.~(\ref{eq: general kernel})) causes 
the variance of $\langle F_{a,\lambda} \rangle$ to be infinite
because $\omega_{a,\lambda}(\theta,s)$ is not
square integrable for any $\lambda$. The singularity is approached as $\epsilon^{-1}$ for $\epsilon \equiv |s-s_0|\rightarrow0$, where $s_0$ is the point where the singularity occurs.
This infinity is simply a mathematical ambiguity and has no physical meaning. It can be avoided using a subsidiary kernel
\begin{equation}
\omega_{a,\lambda,\epsilon}(\theta,s)=\begin{cases}
0 & a<|s|/u(\lambda,\theta)<a+\epsilon\\
\omega_{a,\lambda}(\theta,s) & \textrm{elsewhere},
\end{cases}
\label{eq: subsidiary kernel}
\end{equation}
which corresponds to removing a minimal interval
in the values of the quadrature close to the singularity. This results in a small
systematic error in the estimation of $\langle F_{a,\lambda} \rangle$, equal to
\begin{eqnarray}
\bigtriangleup_{a}(\epsilon )& \equiv &\int_{-\frac{\pi}{2}}^{+\frac{\pi}{2}}\int_{a\, u(\lambda,\theta)}^{\left(a+\epsilon\right)u(\lambda,\theta)} \frac{\mathbf{R}W_{\lambda}\;(\theta,s)}{\pi}\,\left|\omega_{a,\lambda}(\theta,s)\right|\, d\theta ds \nonumber \\
& = & \int_{a}^{a+\epsilon} \mathbf{R}W_{\lambda = 1 }\;(t)\,\left|\omega_{a, \lambda = 1}(t)\right|\, dt.
\label{eq: systematic error}
\end{eqnarray}
Note that this quantity cannot be directly computed from experimental data. Nonetheless, an upper limit on this error can be obtained empirically, by replacing $\mathbf{R}W$ with its
maximum measured value. This leads to an error of the order $O\left(\epsilon^\frac{1}{2}\right)$.
The variance for the estimate of $\langle F_{a,\lambda} \rangle$ is thus finite, but proportional to $|\log\epsilon|$. 
In the following, the use of the the subsidiary kernel $\omega_{a,\lambda,\epsilon}(\theta,s)$ will always be intended even though we will not explicitely write the subscript $\epsilon$.

\subsection{Monte Carlo numerical quadrature}

We now introduce the Monte Carlo numerical quadrature for the double integral appearing in the back-projection formula.
The goal is to  divide the integrand into two parts,  
$\mathbf{R}W_{\lambda}(\theta,s)\omega_{a,\lambda}(\theta,s)/\pi \equiv p_\lambda(\theta,s) \, I_{a,\lambda}(\theta,s)$,
where $p_\lambda(\theta,s)$ is considered to be a probability distribution for two random variables $\theta$ and $s$, and $I_{a,\lambda}(\theta,s)$ is a suitable integrand reminder. Note that in the experiment, only the variable $s$, which characterises the quadrature value, is a truly random variable resulting from the probabilistic nature of the measurement, whereas the phase angle can be controlled experimentally. Nonetheless, we assume that the phase angle is also varied randomly so that, for each shot of the experiment, the experimentalist picks the phase angle according some probability distribution $ \mathcal{C}(\theta) $.
According to Eq.~(\ref{eq: final back projection}), we can then express the filtered back projection as 
\begin{equation}
\langle F_{a,\lambda} \rangle = \int_{-\frac{\pi}{2}}^{+\frac{\pi}{2}}\int_{\mathbb{R}} I_{a,\lambda}(\theta,s) p_\lambda(\theta,s) \, dsd\theta.
%\label{eq: final back projection}
\end{equation}
with
\begin{equation}
\begin{cases}
p_\lambda(\theta,s) \equiv  \mathbf{R}W_\lambda\;(\theta,s)\, \mathcal{C}(\theta) \\
I_{a,\lambda}(\theta,s) \equiv \frac{\omega_{a,\lambda}(\theta,s)}{\pi\,\mathcal{C}(\theta)},
\end{cases}
\label{eq: Montecarlo splitting}
\end{equation}
where $\mathcal{C}(\theta)$ is any distribution of cuts normalized to unity, $i.e.$, $\int \mathcal{C}(\theta) d\theta = 1$.
At this point, all distributions of cuts are logically equivalent and will give the same result for $\langle F_{a,\lambda} \rangle$,  but as we discuss below the choice of distribution will make a difference for the efficiency of the numerical quadrature.

In an experiment, we assume to obtain a set of $M$ outcomes of the form $\{(\theta_i , s_i )\}$,
distributed according to $p_\lambda(\theta,s)$, so that the average of the test function can be estimated as
\begin{equation}
\langle F_{a,\lambda} \rangle \simeq \frac{1}{M}\sum_{i=1} ^M I_{a,\lambda,\epsilon}(\theta_i,s_i).
\end{equation}
The estimate can then be used to check the classicality condition in Eq.~(\ref{eq: violation clas}).
For the statistical significance of the test we, however,  need to include the error in the Monte Carlo quadrature.  The variance of the estimate $\langle F_{a,\lambda} \rangle$ can be expressed as
\begin{eqnarray}
\sigma^2_{F_{a,\lambda}} = \frac{1}{(M-1)} \sum_{i=1}^M \Big(
I_{a,\lambda}(\theta_i,s_i) - \langle F_{a,\lambda} \rangle  \Big)^2 
\label{eq: Montecarlo error}
\end{eqnarray}
Accordingly, the relevant quantity to consider  for the violation is the  ratio
\begin{equation}
\frac{\langle F_{a,\lambda} \rangle}{\sigma_{F_{a,\lambda}}+\Delta_a (\epsilon)},
\label{eq: ratio IR}
\end{equation}
which accounts both for the statistical uncertainty $\sigma_{F_{a,\lambda}}$ and the systematic error $\Delta_a (\epsilon)$.
%Here the symbol $\mathcal{G}_\mathbf{R}$ is intended to stress the explicit dependencies in the above ratio. 
%SINCE THIS SYMBOL IS NOT COMPARABLE TO THE OTHER G, WOULDN'T IT BE BETTER TO USE A DIFFERENT LETTER (OR NOT INTRODUCE IT AT ALL BUT WAIT FOR THE ONE WE REALLY NEED BELOW?)\lu{OK. Note for myself. shall I modify also the Appendix then?}
%, i.e., as opposed to the  characterization $\mathcal{G}$ introduced for the simple test (GIVE EQ NUMBER) which is independent of the number of measurements, the expression for  $\mathcal{G}_\mathbf{R}$ does not have a simple dependence on the number of measurement due to the combination of systematic and statistical errors.
Since we now consider not only the mean but also the statistical uncertainty, such a ratio depends on the distribution of cuts $\mathcal{C}(\theta)$, which 
plays an essential role in the minimisation of $\sigma_{F_{a,\lambda}}$. This is
 discussed below and in Appendix~\ref{sec: B}, where we consider both the squeezed and unsqueezed single photon state.
Now there is no longer a  sensible way to remove the dependence on the number of measurements $M$ from this ratio. In fact, the asymptotic behaviour for large $M$ is no longer  $\sim \sqrt{M}$ as for the simple test
because there is a  the trade off between $\sigma_{F_{a,\lambda}}$, diverging as $\abs{\log\epsilon}$, and $\Delta_a (\epsilon)$ decreasing as $\epsilon^{\frac{1}{2}}$. Therefore it is not convenient to rescale it by $\sqrt{M}$ and take the limit of $M \rightarrow \infty$ as we did for the simple test (see Eq.~(\ref{eq: violation ratio})). In order to achieve a comparison with the simple test discussed above we consider the ratio
\begin{equation}
\mathcal{R}(M) \equiv \frac{\avg{\mathscr{F}}}{\sigma_{\mathscr{F}}}\,\frac{\sigma_{F_{a,\lambda}}+\Delta_a (\epsilon)}{\langle F_{a,\lambda} \rangle}, 
\end{equation}
which describes the ratio between the statistical certainty of the tests, $i.e.$, the ratio of the number of standard deviations with which the tests violate the classicality criterion. A value $\mathcal{R}(M)>1$ ($\mathcal{R}(M)<1$) thus signifies that the simple test (inverse Radon method) has a larger statistical certainty.  

%Similar to the simple test, we could introduce a figure of merit 
%Similar to the simple test, the relevant quantity to classify how well the tests work is
%\begin{equation}
%\mathcal{G}_\mathbf{R}\big[\mathcal{C},a,\lambda, \epsilon\big] \equiv \,\,\lim_{M \rightarrow \infty}\, \frac{\langle F_{a,\lambda} \rangle}{(\sigma_{F_{a,\lambda}}+\Delta_a (\epsilon) )\sqrt{M}}
%\label{eq: FOG IR}
%\end{equation}
%Note the dependence of $\mathcal{G}_\mathbf{R}$ on the distribution of cuts $\mathcal{C}(\theta)$, which
%plays an essential role in the minimization of $\sigma_{F_{a,\lambda}}$
%as discussed below and in Appendix~\ref{sec: B}, where we consider both the squeezed and unsqueezed single photon state.

\subsection{Single photon state.}

First, we consider the back-projection formula for the unsqueezed single photon state.
The rotational invariance of the Wigner function $W(x,p)$ motivates a phase cut distribution of the form
\begin{equation}
\mathcal{C}(\theta) \equiv \frac{1}{\pi},
\label{eq: Montecarlo p single photon}
\end{equation}
which corresponds to a uniform, flat distribution of the phase cuts. 
As shown in  Appendix~\ref{sec: B}, this is the optimal distribution of cuts for minimizing $\sigma_{F_{a,\lambda}}$. 
Consequently, we have $I(\theta,s) \equiv \pi \omega_{a,\lambda}(\theta,s)$.

\begin{figure}[t!]
\includegraphics[scale=0.39]{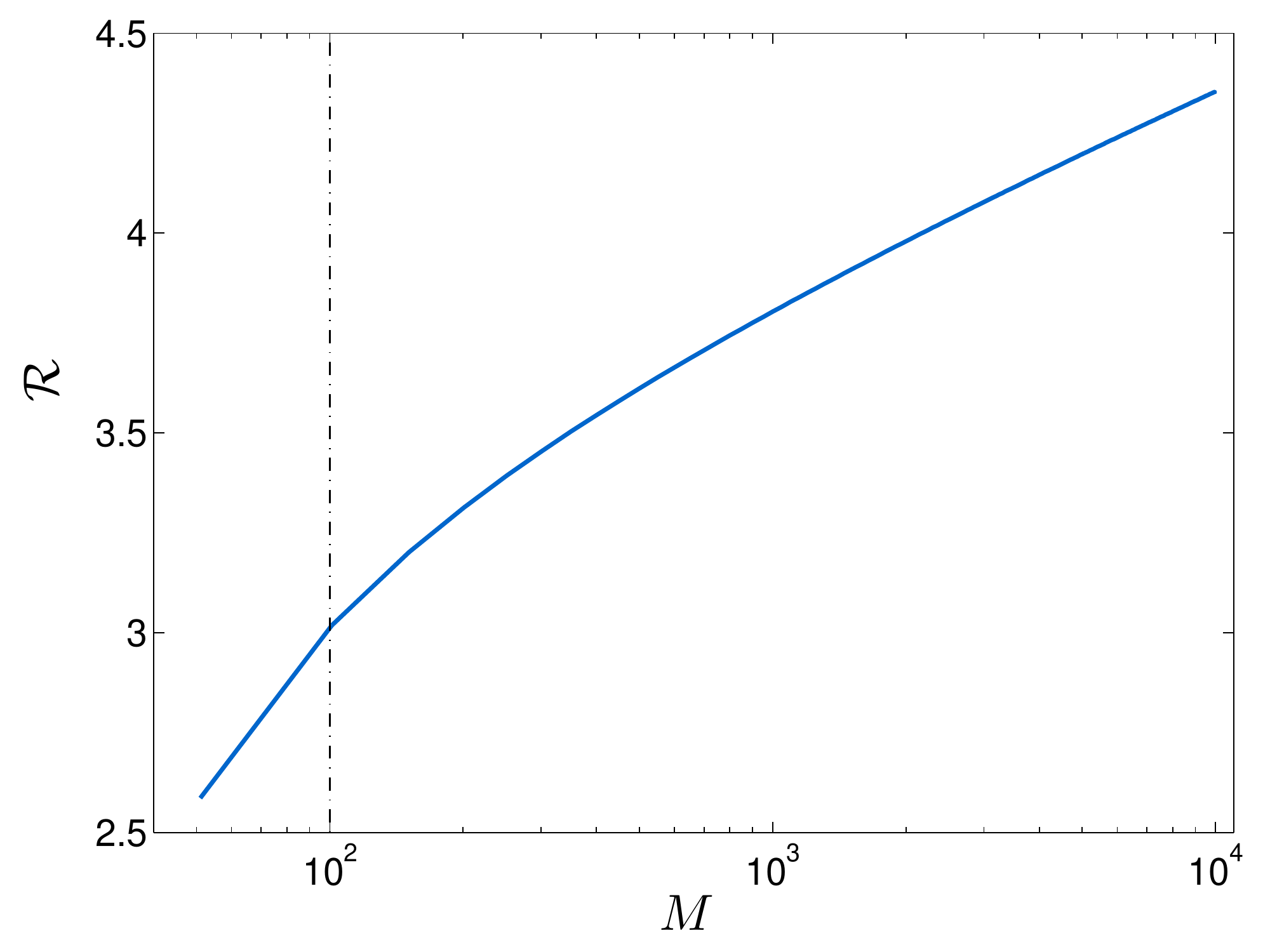}
\caption{Plot of the ratio  $\mathcal{R}(M)$ of the statistical significance of the elementary test and the inverse Radon method for the single photon state.  The comparison is with respect to the best elementary test for $N=16, m=17$. The value $\mathcal{R}(M)>1$ indicates that the elementary test is statistically more reliable. The almost linear part of the graph shows a nearly logarithmic growth for high $M$; the low $M$ region deviates from this behaviour due to the $\sqrt{M-m}$ dependence in $\sigma_{\mathscr{F}}$. The dashed vertical line indicate the  threshold for the inverse Radon method to violate the non-classicality by two standard deviations.}
\label{fig: figure of merit IR}
\end{figure}

For given $M$, we have minimized the value of $\frac{\langle F_{a,\lambda} \rangle}{\sigma_{F_{a,\lambda}}+\Delta_a (\epsilon)}$ with respect to the remaining free parameters $a$ and $\epsilon$, keeping $\lambda=1$ because of the rotational symmetry. In the minimisation, we numerically compute the double integral that estimate the value of the mean $\langle F_{a,\lambda} \rangle$ and variance $\sigma_{F_{a,\lambda}}$ for given $a$ and $\epsilon$ and analytically compute the upper bound of the systematic error $\Delta_a (\epsilon)$. Consequently, we can evaluate the ratio in Eq.~(\ref{eq: ratio IR}) and run a numerical optimisation to minimize it.
We compare the result of the minimisation to the best elementary test we could achieve with $N=16$ ($N$ is the polynomial degree) by calculating the ratio $\mathcal{R}(M)$. The results of this is shown in Fig.~\ref{fig: figure of merit IR} as a function of $M$. From the optimisation we find that the systematic error decreases as a power law and this causes the variance of the kernel to grow logarithmically. Consequently, the quantity $\mathcal{R}(M)$ does not settle to an asymptotic value, but keeps increasing as a function of $M$ because of the trade off between the systematic error and variance of the kernel. For all values of $M$ we find $R>1$ which means that the elementary test is statistically more significant. The inverse Radon is closer to the elementary test for a small number of measurements. Note that the very small $M$ region cannot be investigated since $\sigma_{\mathscr{F}}$ is defined only for $M > m$, the least number of cuts being $m=17$ for $N=16$. 
This region is, however, not very interesting since it lies in the region where the inverse Radon is not capable of producing  a violation by two standard deviations (dashed vertical line in Fig.~\ref{fig: figure of merit IR}).
%I WOULD SAY THERE IS A DIFFERENT REASON: FOR THIS FEW MEASUREMENTS WE PROBABLY DO NOT GET A STATITICALLY SIGNIFICANT VIOLATION. COULD WE PERHAPS MARK A PLACE WHERE WE VIOLATE BY E.G. 2 STANDARD DEVIATIONS?
%In the large $M$ region, instead, the graph asymptotically settles on a precise value ($\sim -0.21$), which by definition is the value of the figure of merit $\mathcal{G}$. Such a convergence shows that the artefice of the subsidiary kernel is not relevent in the end, since for practical situation one always deals with the large $M$ region. In such a practical situation, the confidence of the violation is given by $\frac{\langle F_{a,\lambda} \rangle}{\sigma_{F_{a,\lambda}}+\Delta_a (\epsilon)}$ which is $\sqrt{M}$ times the value of the graph.

\subsection{Squeezed state.}
We now consider the squeezed single photon state, where an explicit optimization in the choice of $\mathcal{C}(\theta)$ is crucial
for obtaining an effective test. 
As shown in Appendix~\ref{sec: B}, the naive approach of choosing $\mathcal{C}(\theta)$ as given by Eq.~(\ref{eq: Montecarlo p single photon}) leads to a poor scaling of the variance as $\sim\lambda^{2}$ $\left(\sim\lambda^{-2}\right)$ for $\lambda\gg1$ ($\lambda\ll1$). This result shows that for the squeezed state a uniform distribution of phase cuts is not desirable since the squeezed state has a highly non-uniform angular distribution. It is thus better to have more measurements in regions of phase space with a rapid variation. To find a more optimal distribution $\mathcal{C}(\theta)$, we use the following result from the theory of Radon transform~\citep{Deans}
\begin{equation}
\mathbf{R}W_{\lambda}\;(\theta,s)=\frac{1}{u(\lambda,\theta)}\mathbf{R}W\;\left(\frac{s}{u(\lambda,\theta)}\right)
\label{eq: squeezing radon transform1}
\end{equation}
together with the transformed kernel  for the stretched filter given in Eq.~(\ref{eq: general kernel}).
Using these two results, we find that
\begin{eqnarray}
\langle F_{a,\lambda} \rangle & = & \int_{-\frac{\pi}{2}}^{+\frac{\pi}{2}}\int_{\mathbb{R}} \mathbf{R}W_\lambda\;(\theta,s)\,\frac{\omega_{a,\lambda}(\theta,s)}{\pi}\, dsd\theta \qquad
\nonumber \\
 & = & \int_{-\frac{\pi}{2}}^{+\frac{\pi}{2}}\frac{1}{ u^2(\lambda,\theta)} \int_{\mathbb{R}}\mathbf{R}W_{\lambda=1}\;(t)\,\frac{\omega_{a,\lambda=1}(t)}{\pi}\, dt d\theta. \nonumber
 \\
 \label{eq: justifying the squeezed quadrature}
\end{eqnarray}
Eq.~(\ref{eq: justifying the squeezed quadrature}) 
shows that not all angles contribute equally to the final result. Because of the use of  a stretched filter most of the weight is focused on a small subset of angles. This is illustrated in Fig.~\ref{fig: g_a many lambdas} where we plot
 $\frac{1}{u^2(\theta,\lambda)}$ as a function of the phase $\theta$ for various values of $\lambda>1$.
\begin{figure}[t!]
\includegraphics[scale=0.42, trim= 5mm 5mm 10mm 5mm]{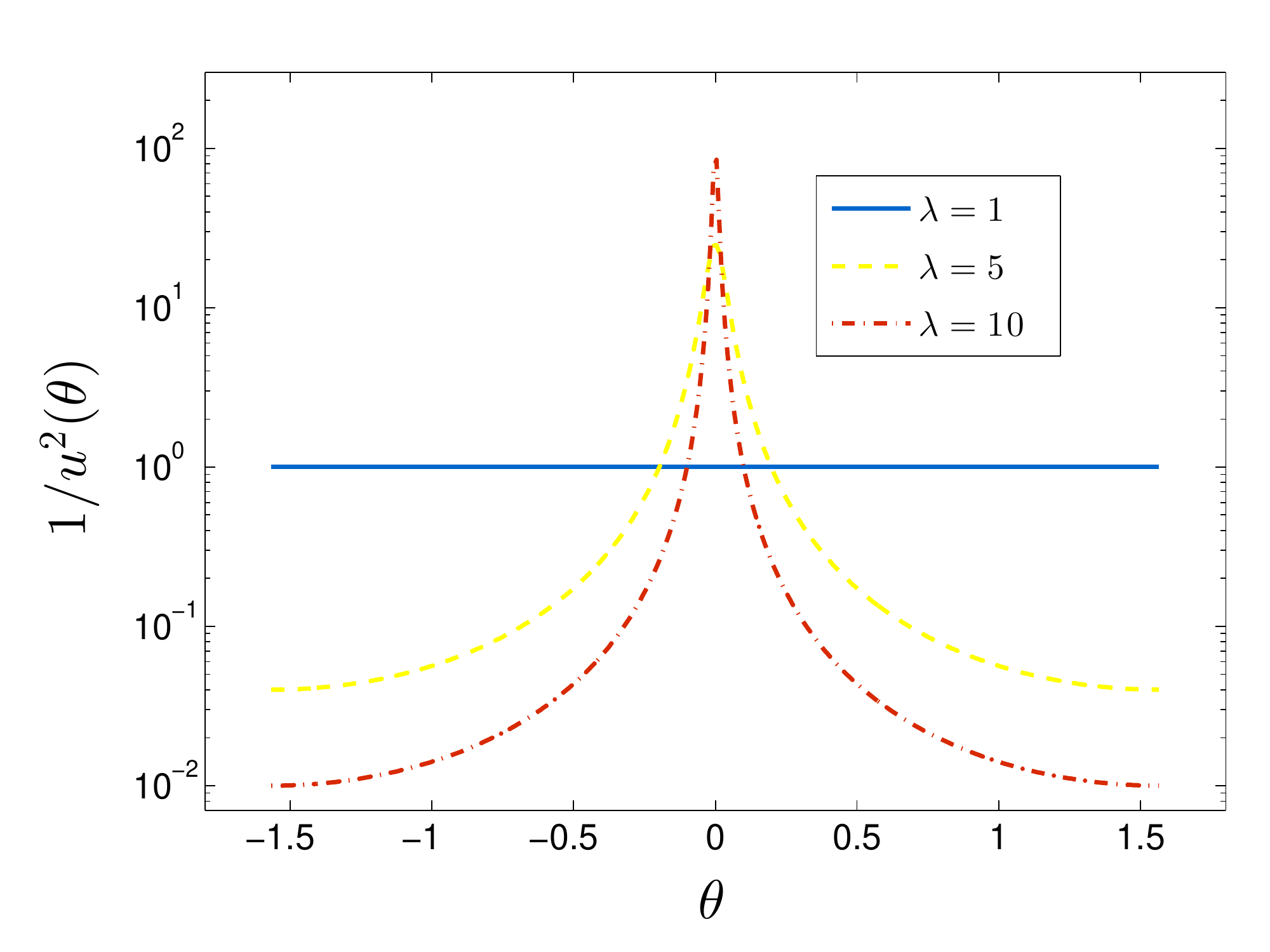}
\caption{Logarithmic plot of $\frac{1}{u^2(\theta,\lambda)}$ as a function of $\theta$ for some values of the squeezing parameter $\lambda$. The region
that contributes the most to the integral in Eq.~(\ref{eq: justifying the squeezed quadrature}) gets narrower
as lambda increases.}
\label{fig: g_a many lambdas}
\end{figure}
As shown in the figure, $\frac{1}{u^2(\theta,\lambda)}$ gives more weight to phases close to
$\theta=0$ ($\theta=\pi$) for $\lambda>1$ ($\lambda<1$). Since it is important to get an accurate sampling of the most relevant part of the integral, it will be advantageous to choose a probability distribution for the phase cuts $\mathcal{C}(\theta)$ which is higher around this region. As shown in Appendix~\ref{sec: B}, the optimal distribution of cuts is indeed
\begin{equation}
\mathcal{C}(\theta) \equiv \frac{1}{\pi u^2(\lambda,\theta)},
\label{eq: Montecarlo p squeezed}
\end{equation}
which coincides with Eq.~(\ref{eq: Montecarlo p single photon}) for $\lambda = 1$.
Furthermore, we show in the appendix that the relevant ratio $\frac{\langle F_{a,\lambda} \rangle}{\sigma_{F_{a,\lambda}}+\Delta_a (\epsilon)}$ does not depend on $\lambda$ for this choice of phase cut distribution, $i.e.$, it is thus equal to the one found for the rotationally symmetric state. This is completely equivalent to the analogous result found for the simple test. Hence the same conclusions drawn about the statistical significance of the elementary and inverse Radon tests for the single photon state also applies to the squeezed single photon state.

\subsection{Finite Number of Cuts}

In experiments, a strategy more widely used than the Monte Carlo quadrature is to have a small set of cuts but many measurements for each cut. 
In order to  study this situation, it is more convenient 
to work out the right hand side of Eq.~(\ref{eq: filtered backprojection}) in a different way than for the Monte Carlo quadrature. This is done in two steps. First, we introduce the function
\begin{equation}
g_{a,\lambda}\left(\theta\right)\equiv\int_{\mathbb{R}}\mathbf{R}W_{\lambda}\;(\theta,s)\,\omega_{a,\lambda}(\theta,s)\, ds,
\label{eq: g angular function}
\end{equation}
which is assumed to be accurately estimated by having many measurements per cut.
The integral over the phase cuts is then written as a summation over a convenient choice
of cuts $\{\theta_{j}\}$ with weights $w_{j}$ 
\begin{equation}
\frac{1}{\pi}\int_{-\frac{\pi}{2}} ^{+\frac{\pi}{2}} g_{a,\lambda}\left(\theta\right)d\theta  \quad\rightarrow\quad\sum_{\{\theta_{j}\}}w_{j} \: g_{a,\lambda}\left(\theta_{j}\right).
\label{eq: numerical quadrature}
\end{equation}

The discretisation of the integral above, formally called numerical quadrature, is the crucial step in   a strategy with a finite number of cuts.
This choice, however, gives rise to ambiguities because we cannot completely know the function
$ g_{a,\lambda}\left(\theta\right)$ since it is sampled with only a few phase arguments $\{\theta_{j}\}$.
The first ambiguity is in the choice of the number of phase cuts. We cannot $a$ $priori$ determine a minimum number of phase cuts required as in the elementary test.
Regardless of the number of cuts, the full expression on the left hand side of Eq.~(\ref{eq: numerical quadrature}) involves an integral. Hence one can always wonder if adding more phase cuts could reveal narrow features missed by the finite number of cuts. It is therefore in principle necessary to constrain the influence of such narrow features on the final result to convincingly demonstrate non-classicality. 
Another ambiguity is in the estimation of the error due to
the numerical quadrature, which cannot be easily determined unless some other information about $ g_{a,\lambda}\left(\theta\right)$ is known. 

Despite of these ambiguities, it is still reasonable to ask whether we can determine an efficient
numerical quadrature that uses a low number of phase cuts compared to the total number of measurements.
In other words, we want  to understand whether the inverse Radon method can provide
a way to choose a suitable small set of phase cuts if some further sensible assumptions are made.
To give an answer to this question, in the analysis carried out in Appendix~\ref{sec: B},
we assume $ab$ $initio$ 
that the states considered are the single photon state, described by a symmetric $W_\lambda(x,p)$, as well as the squeezed single photon state.  For our theoretical analysis, such an hypothesis relies ultimately on the Wigner function to be the correct description of the state, which is somewhat logically problematic since this information is exploited to carry out the method, not just to optimise it. 

We show in Appendix~\ref{sec: B} that this strategy can also be carried out in such a way that the statistical
strength of the test is the same for both rotationally symmetric and squeezed states. 
As the squeezing increases, however, more cuts are needed in order keep the uncertainty of the test below a certain threshold and these cuts should be highly asymmetrically distributed. Choosing such an asymmetric  distribution can be justified by the asymmetry of the filter function, which is highly peaked around a narrow distribution of angles. 
On the other hand, it is hard to quantitatively account for the absence of further errors due to the unknown angular dependence of the underlying physical state, and this introduces an ambiguity in the test.
%but still it would be desirable not to have to justify quantitatively that such a choice do not artificially allow for errors due to an angular dependence of the underlying physical state. 

In essence, closing such \qts{loop-holes} in the argument is what is achieved with the elementary test, which is explicitly designed to deal with a finite number of cuts. Due to the ambiguities discussed here, making a more direct comparison between the elementary test and a strategy with a finite number of cuts is highly non-trivial and is thus beyond the scope of this article.

\section{Conclusion and discussion}

In conclusion, we have generalised the non-classicality test presented in Ref.~\citep{Kot12}  and developed a framework for finding the optimal test function for a given quantum state. In the generalisation of the test, we consider test functions which are arbitrary polynomials in the phase space. Thereby, we allow  test functions which do not have rotational symmetry and are therefore more suited for testing the non-classicality of quantum states which do not possess this symmetry. As a  specific example, we have found the optimal test function for a squeezed single photon state and shown that it results from a simple transformation of the rotationally invariant test function considered in Ref.~\citep{Kot12}.  

Furthermore, we have performed a comparison between the simple non-classicality test and more conventional tests based on the inverse Radon transform. 
In order for the comparison to be fair, we have extended the strategy of inverse Radon transformation based on the filtered back-projection formula so that it is more optimised  for squeezed states. %We have implemented the filtered back-projection algorithm for a suitable set of non-negative filters here considered as test functions. 
The constraints of the filtered back-projection formula and the limited number of known completely positive filter functions meant, however, that only a rather restricted class of test functions could be investigated (as compared with the elementary test, where there is a lot of freedom). In addition, the filtered back-projection formula
needs to be handled with care because it introduces problems with divergencies and ambiguities from introducing a finite number of cuts. With a Monte Carlo quadrature method the  ambiguities from the finite number of cuts can be avoided, but this requires the use of as many quadrature cuts as the number of measurements. Reducing the number of cuts introduces undesired ambiguities, which may question the whole outcome of the test. In particular, these ambiguities are severe for strongly squeezed states where it is necessary to have a large number of cuts in a small angular interval where the integrand varies rapidly. Without any prior information about the state, such a choice of distribution introduces ambiguities, which make it hard to give a convincing proof of the non-classicality of a state. All of these problems are avoided with the elementary test. In addition to this, we find that the elementary test is statistically stronger than the methods based on the inverse Radon transform (at least compared to the version based on the Monte Carlo quadrature method). 

We thus find that the simple test is  more suited and straightforward than the inverse Radon transform for experimental demonstration of non-classicality. 
The advantage of the simple test originates from the fact that it is naturally designed as a non-classicality test under the constraints imposed in a typical experiment. In particular, the elementary test is specifically designed to work with a finite number of quadratures being measured. On the contrary, the
inverse Radon transform is a more general method, which allows to gain more information about the physical state but it is less suited for this specific situation. The statistical weakness of the latter could, however, be due to the fact that the kernels considered here are not ideal due to their discontinuity. On the other hand, the smooth filters considered in the literature that grant a better convergence are not necessarily completely positive. A negativity observed with such  a filter would therefore inevitably introduce the question of whether the  result was caused by the negativity of the filter and this could compromise the test.

We have considered the simple test for a quadrature squeezed single photon state but the general framework developed here can be applied to any quantum state or any physical system. In particular, it could be intersting to use the test to demonstrate non-classicality of opto-mechanical systems~\citep{Optomechanics10, Optomechanics14}. Currently much effort is devoted to reaching the quantum regime for these systems. To convincingly show that such systems are outside the realm of classical physics, it is essential to demonstrate a contradiction with classical physics along the lines discussed here. Also it could be interesting to apply the method to investigate, the negativity  of macroscopic quantum states such as large Schr{\"o}dinger cat states and potentially use this as a measure of non-classicality in such systems.  
Finally, one could imagine using  the elementary test method  as a general tool for tomography without resorting to  the heavy machinery of inverse Radon transformations. Since we find the method to be more reliable and less ambiguous, one could imagine other situations where this method could be a desirable strategy for tomography.

\begin{acknowledgments}
We thank Eran Kot for valuable discussions and gratefully acknowledge financial support from the Lundbeck foundation, the Carlsberg foundation and the Famiglia Legnanese foundation. The research leading to these results was funded by the European Union Seventh Framework Programme through  ERC Grant QIOS (Grant No. 306576). 
\end{acknowledgments}

\appendix

\section{Test function method}
\label{sec: A}

In this appendix, we present the details of the results on the simple test method discussed in the main text. First, we prove,
in three intermediate steps, a slightly more general version of Theorem~\ref{thm: 1}.
The first step is to establish that for a rotationally symmetric distribution the average value of any test function can be expressed as the average value of a radial test function.
The second step is to identify the tests
that minimise the variance for a given value of the mean.
Finally, the explicit construction of a radial test which minimises the variance is provided. The proof of the theorem is then obtained by
combining these three partial results.
Second, we prove Theorem~\ref{thm: 2} concerning
squeezed states.

\subsection{General formulation.}
We will take into account a slightly more general form of the test function than in the main text 
in order to derive the proofs of the theorems. To begin, we consider the test function to have the form
\begin{equation}
\mathscr{F}(x,p)=1+\sum_{i=1}^{2N}P_{i}(x,p)
\label{eq: formal general test function}
\end{equation}
where  $P_{i}(x,p)$ is 
a homogeneous polynomial of degree $i$ defined on the phase space of $x$ and $p$.
Moreover, all the test functions of the form in Eq.~(\ref{eq: formal general test function}) are implicitly taken to be non-negative, thus we are more general than in the main text where we only considered square polynomials.
Note that the polynomials $P_{i}$ can be expressed in terms of the functions $x^{i-k}p^{k}$ to recover the notation employed in the main text of the paper.
By assumption, we expect $\mathscr{F}$ to be in the algebra of observables,
formally defined as the commutative algebra generated by all the quadratures.
In particular, for fixed $i$, any $P_{i}(x,p)$ belongs to a linear subspace of this algebra,
together with $Q_{\theta}^i$ for any $\theta$. Any set of at least $(i+1)$ functions
of the form
$Q_{\theta}^i$ for different $\theta$ is a (possibly overcomplete) set in this linear subspace.
Accordingly, given a set of $j=1,..,m$ different cuts $\{\theta_j\}$, with $m\geq2N+1$, 
all the polynomials 
$P_{i}(x,p)$ involved in Eq.~(\ref{eq: formal general test function})
can be expressed as
\begin{equation}
P_{i}(x,p)=\sum_{j=1}^{2N+1}T_{i;j}Q_{\theta_{j}}^{i},
\label{eq: general linear combination}
\end{equation}
with $T_{i;j}$ being real-valued but not uniquely defined. We will express a test, $\mathbb{T}$ as
\begin{equation}
\mathbb{T}\equiv\left(N,\: P_{i},\:\theta_{j},\: T_{i;j}\right),
\label{eq: formal explicit test}
\end{equation}
where the non-negativity of the test function is intended, but not explicitly stated. This is
the most general class of tests we consider and we refer to this class  as
$\mathscr{C}$. A test of the form in Eq.~(\ref{eq: formal explicit test})
belongs to this class, $\mathbb{T}\in\mathscr{C}$, if $\mathscr{F}$
as defined in Eq.~(\ref{eq: formal general test function}) is non-negative
and Eq.~(\ref{eq: general linear combination}) holds.

The rotated coordinates, $Q_{\phi}$ and $P_{\phi}$ are defined as
\begin{equation}
\left(\begin{array}{c}
Q_{\phi}\\
P_{\phi}
\end{array}\right)=\left(\begin{array}{cc}
\cos\phi & \sin\phi\\
-\sin\phi & \cos\phi
\end{array}\right)\left(\begin{array}{c}
x\\
p
\end{array}\right)\label{eq: rotation of coordinate}
\end{equation}
and correspond to the quadrature $Q_{\phi}$ and the rotated
canonical momentum $P_{\phi}$. Accordingly, we define
the rotation of polynomials by
\begin{equation}
P_{i}[\phi](x,p)\equiv P_{i}(Q_{\phi},P_{\phi}).
\end{equation}
We refer to a test as a radial test and denote it by $\mathbb{T}_{rad}$, if $P_{i}[\phi](x,p)= P_{i}(x,p)\quad\forall\phi\in\mathbb{R}$, for all $ i=1,..,2N+1$. Consequently, in this case, the test function, $\mathscr{F}$, is only a function of the radius. The class of such radial tests is denoted $\mathscr{C}_{rad}$ and it is obviously a proper subclass of $\mathscr{C}$.
To complete the set of notations, it is sensible to explicitly express the dependence of ${\langle\mathscr{F}\rangle}/{\sigma_{\mathscr{F}}}$ in terms of the free parameters of the test. Therefore we shall write
\begin{equation}
\frac{\langle\mathscr{F}\rangle}{\sigma_{\mathscr{F}}} \equiv  \mathcal{G}[\mathbb{T},M],
\end{equation}
so that the actual figure of merit $\mathcal{G}$ is simply recovered as $\mathcal{G}[\mathbb{T}] = \lim_{M \rightarrow \infty} \frac{\mathcal{G}[\mathbb{T},M]}{\sqrt{M-m}}$

\subsection{Rotationally invariant states.}
As hinted above, even though our objective is the optimal $\mathcal{G}[\mathbb{T},M]$ for any $M$, we do not yet consider the direct optimisation with respect to the free parameters of the test. Instead, we first want to understand the structure of the class $\mathscr{C}$
and its relation to $\mathscr{C}_{rad}$ for a rotationally symmetric $W(x,p)$.
The first result concerns the mean value $\langle\mathscr{F}\rangle$ and, in particular
how the class $\mathscr{C}$ and $\mathscr{C}_{rad}$ are related
under the condition that all $\left\langle P_{i}\right\rangle $ are given.

\begin{prop}
Assume that $W(x,p)$ is rotationally symmetric 
and consider a generic test $\mathbb{T}\in\mathscr{C}$ such that 
$\mathscr{F}=1+\sum_{i=1}^{2N}P_{i}(x,p)$ is the test function.
There always exists a radial test $\mathbb{T}_{rad}\in\mathscr{C}_{rad}$
with test function of the form $\mathscr{F}_{rad}=1+\sum_{i=1}^{2N}\bar{P}_{i}(x,p)$,
where $\bar{P}_{i}(x,p)$ is rotationally symmetric and fulfils $\langle \bar{P}_{i}\rangle = \langle P_{i}\rangle$ for any $i$,
so that $\langle\mathscr{F}\rangle=\langle\mathscr{F}_{rad}\rangle$.
Furthermore, all the radial tests that
satisfy this condition differ from each other only in the parameters
$\left\{ \theta_{j}\right\} ,\:\left\{ T_{i;j}\right\} $. 
\label{pr: Prop A1}
\end{prop}

\begin{proof}
The proof is by construction. The given test is characterized by $\mathbb{T}\equiv\left(N,\: P_{i},\:\theta_{j},\: T_{i;j}\right)$,
and we define
\begin{equation}
\bar{P}_{i}\equiv\frac{1}{2\pi}\int_{0}^{2\pi}P_{i}\left[\theta\right]d\theta \qquad\forall i=1,..,2N.
\end{equation}
Note that $\bar{P}_{i}$ 
is rotationally invariant since
\begin{equation}
\bar{P}_{i}\left[\phi\right]=	\frac{1}{2\pi}\int_{0}^{2\pi}P_{i}\left[\theta+\phi\right]\textrm{d}\theta=\bar{P}_{i}.
\end{equation}
Moreover, because of the rotational symmetry of $W(x,p)$, a simple
change of variables to the rotated coordinates shows that $\left\langle P_{i}\left[\theta\right]\right\rangle =\left\langle P_{i}\right\rangle $
for any $\theta$. Thus,
\begin{eqnarray}
\left\langle \bar{P}_{i}\right\rangle  & = & \frac{1}{2\pi}\int_{0}^{2\pi}\left\langle P_{i}\left[\theta\right]\right\rangle \textrm{d\ensuremath{\theta}}\nonumber \\
 & = & \left\langle P_{i}\right\rangle \frac{1}{2\pi}\int_{0}^{2\pi} d\theta=\left\langle P_{i}\right\rangle .
\end{eqnarray}
Considering $\mathscr{F}_{rad}=1+\sum_{i=1}^{2N}\bar{P}_{i}(x,p)$, it can be seen that the rotationally
invariant test function is non-negative since
\begin{eqnarray}
\mathscr{F}_{rad}(x,p) & = & 1 + \sum_{i=1}^{2N}\frac{1}{2\pi}\int_{0}^{2\pi}P_{i}\left[\theta\right](x,p) d\theta
\nonumber \\
 & = & \frac{1}{2\pi}\int_{0}^{2\pi} \mathscr{F}\left[\theta\right](x,p) d\theta,
\end{eqnarray}
$i.e.$, for fixed $(x,p)$, it is the integral  
of a non-negative function $\mathscr{F}_{rad}\left[\theta\right](x,p)$ wrt.~the variable $\theta$. The statement of Proposition~\ref{pr: Prop A1} follows immediately. Note that no consideration of $\theta_{j}$ or $T_{i;j}$ was necessary because they are irrelevant for the mean value of $\mathscr{F}$. 

\end{proof}
Above, we have proven that for a rotationally symmetric distribution, $W(x,p)$, and a non-rotational test, we can always find a rotationally symmetric version of the test, which gives the same average. We will now show that the rotationally invariant test will also have a smaller uncertainty in the mean value.
To do so, we initially study a more general problem, neglecting for now the non-negative property of the test functions. Since we consider a more general class of functions, 
we will at first (see Proposition~\ref{pr: A2}) achieve results without considering the requirements for proper test functions, and thereby of no direct physical significance. Later we will, however, show (see Proposition~\johs{\ref{pr: A3}}) that it is indeed possible to find a a physically meaningful test fulfilling the conditions of a optimum in the broader class, and thus also an optimum in the restricted class. 

First, we give a suitable expression for $\sigma^2_{\mathscr{F}}$ that accounts
for the rotational symmetry of $W(x,p)$. Since $\langle Q^i_{\theta} \rangle \equiv
\langle Q^i \rangle$ does not depend on $\theta$ for any $i$, we shall introduce
the covariance matrix of the quadratures,
\begin{equation}
\frac{1}{2}g^{ii'} \equiv \left\langle Q^{i+i'}\right\rangle -\Bigl\langle Q^{i}\Bigr\rangle\left\langle Q^{i'}\right\rangle 
\label{eq: cov matrix quadratures}
\end{equation}
It can then be seen  that $\sigma^2 _j \equiv \langle \mathscr{H}^2_j\rangle-\langle \mathscr{H}_j\rangle^2 = \frac{1}{2}\sum_{i, i'=1}^{2N}g^{ii'} T_{i;j}T_{i';j}$ so that a final expression for $\sigma^2_{\mathscr{F}}$ is achieved
\begin{equation}
\sigma^2_{\mathscr{F}} = \frac{1}{2}\sum_{j=1}^{m}\sum_{i, i'=1}^{2N}g^{ii'} \frac{T_{i;j}T_{i';j}}{M_j-1},
\label{eq: append expression err}
\end{equation}
where $\{M_j\}$ is the number of measurements for the $j$'th cut and $\sum_j M_j = M$.

We now perform an optimisation in the class of functions
of the form $\mathscr{F}=1+\sum_{i,j}T_{i;j}Q_{\theta_{j}}^{i}$, with the only constraint
$\sum_{j}T_{i;j}=K_{i}$.  This constraint on $T_{i;j}$ originates from the identity in Eq.~(\ref{eq: general linear combination}). In the case of rotational symmetry, where $\left\langle Q_{\theta_{j}}^{i}\right\rangle \equiv\Bigl\langle Q^{i}\Bigr\rangle$
does not depend on $\theta_{j}$, this implies that
\begin{equation}
\sum_{j=1}^{m}T_{i;j}=\frac{\left\langle P_{i}\right\rangle }{\left\langle Q^{i}\right\rangle }\equiv K_{i},
\label{eq: append constraint Tij}
\end{equation}
where $K_{i}$ are constants that depend only on the values $\left\langle P_{i}\right\rangle$
and no other free parameters of the test. 
By optimising the variance subject to this constraint, we thus optimise over the class of tests giving a fixed value of the mean value. Since the true optimum must also be the optimum for the particular mean value obtained, properties proved under this constraint also applies to the true optimum. The optimisation is done keeping $N$, $M$ and $m$ fixed, since these parameters do not depend on whether the test is radial or general. We can thus state and prove the following result.
\begin{prop}
Assume $\mathscr{F}$ to be any generic polynomial function 
of order $2N$ over an arbitrary set of $m$ quadratures, that is 
\begin{equation}
\mathscr{F}=1+\sum_{j=1}^{m}\sum_{i=1}^{2N}T_{i;j}Q_{\theta_{j}}^{i}
\end{equation}
We assume that all the quadratures $Q_{\theta_{j}}$ are independent, identically distributed
random variables and that $T_{i;j}$ are only constrained
by the relation $\sum_{j=1}^{m}T_{i;j}=K_{i}$.
It follows that the minimum value of $\sigma^2_{\mathscr{F}}$, after optimising $T_{i;j}$
and keeping fixed $N$, $M$ and $m$, is
\begin{equation}
\sigma^2_{\mathscr{F}} = \frac{1}{2}\frac{\sum_{i, i'=1}^{2N}g^{ii'} K_i K_{i'}}{M-m}
\label{eq: append complete minimum err}
\end{equation}
and is thus independent of the particular choice of $\{M_j\}$ satisfying $\sum_j M_j = M$.
Furthermore, if $\{M_j\}$ are all the same, the best
$T_{i;j}$ is independent of $j$, $i.e.$, $T_{i;j}=K_{i}/m$, $m$ being the number of cuts.
\label{pr: A2}
\end{prop}

\begin{proof} 
Consider a distribution of measurements $\{M_j\}$. For simplicity, we define $\tilde{M}_j = M_j -1$ so that $\sum_j \tilde{M}_j = M-m$.
The constrained optimisation problem is then studied using the Lagrangian function
\begin{multline}
\mathscr{L}\left(\left\{ T_{i;j}\right\} ;\left\{ \lambda^{i}\right\} \right)= \sigma_{\mathscr{F}}^2 -\sum_{i=1}^{2N}\lambda^{i}\sum_{j=1}^{m}T_{i;j} \\
 =\sum_{j=1}^{m}\left\{ \frac{1}{2}\sum_{i,i'=1}^{2N} \frac{g^{ii'}T_{i;j}T_{i';j}}{\tilde{M}_j} -\sum_{i=1}^{2N}\lambda^{i}T_{i;j}\right\}, 
 \label{eq: lagrangian}
\end{multline}
which only depends parametrically on $\{\tilde{M}_j\}$.
The optimal solution is determined by 
\begin{equation}
\begin{cases}
\frac{\partial\mathscr{L}}{\partial T_{i;j}}=\sum_{i'=1}^{2N} \frac{g^{ii'}T_{i';j}}{\tilde{M}_j} -\lambda_{i}=0 & \forall i,j\\
\sum_{j=1}^{m}T_{i;j}-K_{i}=0 & \forall i \\
\sum_{j=1}^{m} \tilde{M}_j= M-m.
\end{cases}
\end{equation}
It is seen from the first equation that the optimal solution is such that $T_{i;j}/\tilde{M}_j$ does not depend on $j$, $i.e.$, it is independent of the choice of phase cuts. From this, we find that the solution fulfilling all three conditions is
\begin{equation}
T_{i;j} = \frac{K_i \, \tilde{M}_j}{M-m}.
\label{eq: arg min T radial}
\end{equation}
Then, the constrained stationary value of the function in Eq.~(\ref{eq: lagrangian}) is
\begin{equation}
\sigma_{\mathscr{F}}^{2} = \frac{1}{2}\frac{\sum_{i,i'=1}^{2N}g^{ii'}K_{i}K_{i'}}{M-m}.
\label{eq: best error radial}
\end{equation}
It can be proven that it is a minimum by noting that the Hessian of Eq.~(\ref{eq: append complete minimum err}) is the tensor product of the covariance matrix and the diagonal matrix containing $1/\tilde{M}_j$ on the diagonal. Since both of these matrices are positive-semidefinite by definition the found extremum is a minimum.
%We note that this is an actual minimum because the Hessian matrix of $\sigma^2_{\mathscr{F}}$ in Eq. (\ref{eq: append expression err}) is
%the covariance matrix $g$, which is a positive matrix by definition (null eigenvalues are excluded because the
%state is not dispersion free, $i.e.$, all the entries of $g$ are strictly positive).
Finally, it follows from Eq.~(\ref{eq: arg min T radial}) that when
$\{M_j\}$ are all the same, the ratio $\frac{\tilde{M}_j}{M-m}$ is exactly $1/m$.

\end{proof}
The crucial part of the theorem 
is the independence of the covariance matrix in Eq.~(\ref{eq: cov matrix quadratures}) on the choice of phase cuts,
which is a direct consequence of the rotational symmetry of $W(x,p)$.
It follows that the choice of the cuts is irrelevant for
 Proposition~\ref{pr: A2}. 
Furthermore, the minimum of $\sigma_{\mathscr{F}}^{2}$ does not depend on the choice of $\{M_j\}$, provided that $M$ is fixed. This is in contradiction with the physical intuition that it is  not a suitable strategy to perform most of the measurements on a single quadrature whereas  only a few measurements are performed on the remaining quadratures.  
Note, however, that this result above is obtained in the broader class of functions only subject to the constraint   $\sum_{j}T_{i;j}=K_{i}$. 
Hence, for arbitrary choice of $\{M_j\}$, we might not be able to construct radial polynomials fulfilling the general requirement Eq.~(\ref{eq: general linear combination}) with the optimal $T_{i;j}$ given by Proposition~\ref{pr: A2}. As we will show below, it is, however, possible to find a proper test fulfilling both Eq.~(\ref{eq: general linear combination}) and the optimum in Proposition~\ref{pr: A2}. To this end, it is sensible to fix $M_j$ to be all the same, so that the statistics per 
cut is unbiased. Then we will show that we can construct radial polynomials 
according to Eq.~(\ref{eq: general linear combination}) in such a way that $T_{i,j}$
is independent of $j$. Note that 
the only radial polynomial we can construct must be functions of $(x^{2}+p^{2})$ to retain rotational symmetry. 

\begin{prop}
Assume $T_{i;j}$ to be independent of $j$. Then the polynomials $P_{i}$
given by the identity in Eq.~(\ref{eq: general linear combination}) are radial
polynomials provided that the phase cuts are uniformly distributed. In particular,
the identity 
\begin{equation}
\left(x^{2}+p^{2}\right)^{n}=T_{2n}\sum_{j=1}^{m}Q_{\theta_{j}}^{2n}
\label{eq: cuts uniform proof}
\end{equation}
holds, where $T_{2n}=\left(\begin{array}{c}
2n\\
n
\end{array}\right)^{-1}\frac{4^{n}}{m}$ and $\theta_{j}=\frac{j\pi}{m}$ with $m\geq n+1$. 
\label{pr: A3}
\end{prop}
Note that we
have dropped the $j$ index in the coefficient $T_{i,j}\Bigl|_{i=2n}\equiv T_{2n}$ and
 we have used the index $n$ in place of $i$ to avoid confusion
in the subsequent proof of the statement, where the  imaginary unit
is used.
\begin{proof}
Since $T_{i;j}$ do not depend on $j$, we consider the following
sum
\begin{equation}
\sum_{j=1}^{m}Q_{\theta_{j}}^{2n}
\end{equation}
where for the sake of generality, the sum is extended over $m$ cuts, with 
the suitable $m$ to be determined. Using
the exponential representation of the sine and cosine functions, we
get
\begin{eqnarray}
Q_{\theta_{j}}^{2n} = \frac{1}{4^{n}}\left(\textrm{e}^{i\theta_{j}}\left(x-ip\right)+\textrm{e}^{-i\theta_{j}}\left(x+ip\right)\right)^{2n} \qquad \qquad & \\ \nonumber
 \quad =\frac{1}{4^{n}}\sum_{l=0}^{2n}\left(\begin{array}{c}
2n\\
l
\end{array}\right)\left(x-ip\right)^{2n-l}\left(x+ip\right)^{l}\textrm{e}^{i2\theta_{j}(n-l)}&
\label{eq: esponential representation in quadraturess}
\end{eqnarray}
so that 
\begin{multline}
\sum_{j=1}^{m}Q_{\theta_{j}}^{2n}= \\
=\frac{1}{4^{n}}\sum_{l=0}^{2n}\left(\begin{array}{c}
2n\\
l
\end{array}\right)\left(x-ip\right)^{2n-l}\left(x+ip\right)^{l}\left(\sum_{j=1}^{m}\textrm{e}^{i2\theta_{j}(n-l)}\right)
\label{eq:resummation radial quadratures}.
\end{multline}
By choosing uniformly distributed cuts, we can turn the summation over
$j$ into a simple geometric series. Letting $\theta_{j}=\frac{j\pi}{m}$, the geometric series is 
\begin{equation}
\sum_{j=1}^{m}\left(\textrm{e}^{i\frac{(n-l)}{m}2\pi}\right)^{j}
\end{equation}
and the exponential is different from unity if $\frac{n-l}{m}$ is not an
integer. Accordingly, if $m\geq n+1$, we have that
\begin{equation}
\sum_{j=1}^{m}\textrm{e}^{i2\theta_{j}(n-l)}=m\delta_{l,n}
\label{eq: geometric series cuts}
\end{equation}
Inserting this into Eq.~(\ref{eq:resummation radial quadratures}) gives
\begin{equation}
\sum_{j=1}^{m}Q_{\theta_{j}}^{2n}=\frac{m}{4^{n}}\left(\begin{array}{c}
2n\\
n
\end{array}\right)\left(x^{2}+p^{2}\right)^{n}
\end{equation}
and hence the statement of the proposition is proven.
\end{proof}
Note that
only $n+1$ cuts are needed to express a polynomial of degree $2n$
such as $\left(x^{2}+p^{2}\right)^{n}$. This is because odd powers of $x$ and $p$ are not involved. 
Therefore, for a radial test function of degree $2N$ (corresponding to the square of a polynomium of degree $N$), we only need a minimum number of
$N+1$ cuts. 

We now have all the intermediate results necessary to prove the statement of Theorem~\ref{thm: A1},
which constitues the main result concerning the optimal $\mathcal{G}[\mathbb{T},M]$.

\begin{thm}
Assume that $W(x,p)$ is rotationally symmetric and takes
on negative values; then the optimal test, $i.e.$, the test $\mathbb{T}_{best}$
that minimises $\mathcal{G}[\mathbb{T},M]$ for any $M$ belongs to the subclass $\mathscr{C}_{rad}$
and the cuts are uniformly distributed on the plane (see Eq.~(\ref{eq: cuts uniform proof})).
\label{thm: A1}
\end{thm}

\begin{proof}
Let us fix $M$ and pick up a generic test $\mathbb{T}\in\mathscr{C}$ characterized
by $\left(N,\: P_{i},\:\theta_{j},\: T_{i;j}\right)$, for which $\mathcal{G}[\mathbb{T},M]$ is negative. 
We will construct a radial test $\mathbb{T}_{rad}\in\mathscr{C}_{rad}$ for which $\mathcal{G}[\mathbb{T}_{rad},M]\leq\mathcal{G}[\mathbb{T},M]$.

By Proposition~\ref{pr: Prop A1}, we can find a set of radial polynomials $\bar{P}_{i}$
such that $\langle \bar{P}_{i}\rangle =\langle P_{i}\rangle$
and clearly $\Bigl\langle\mathscr{F}\Bigr\rangle=\Bigl\langle\mathscr{F}_{rad}\Bigr\rangle$.
By symmetry considerations, we know that only radial polynomials with
even degree will contribute to the mean value. 
We now put $\bar{T}_{i;j}\equiv\bar{T}_{2n}=\frac{\left\langle P_{2n}\right\rangle }{m\left\langle Q^{2n}\right\rangle }$.
By Proposition~\ref{pr: A3}, we pick up a set of uniformly distributed cuts, $\bar{\theta}_{j}=\frac{j\pi}{m}$
with $m$ being, for simplicity, the number of cuts in $\mathbb{T}$, so that the
identity in Eq.~(\ref{eq: general linear combination}) holds, $i.e.$, $\bar{P}_{2n}=\bar{T}_{2n}\sum_{j}Q_{\bar{\theta}_j}^{2n}$.
Furthermore, it is seen that both
$T_{i;j}$ and $\bar{T}_{i;j}$ satisfy the same constraint
\begin{equation}
\sum_{j=1}^{m}T_{i;j}=\sum_{j=1}^{m}\bar{T}_{i;j}=\frac{\left\langle P_{2n}\right\rangle }{\left\langle Q^{2n}\right\rangle }.
\end{equation}
Consequently, both $\mathscr{F}=1+\sum_{i,j}T_{i;j}Q^i_{\theta_j}$
and $\mathscr{F}_{rad}=1+\sum_{i,j}\bar{T}_{i;j}Q^i_{\bar{\theta}_j}$ 
satisfy the hypothesis of Proposition~\ref{pr: A2}
with the same $K_i = \frac{\left\langle P_{2n}\right\rangle }{\left\langle Q^{2n}\right\rangle }$.
It follows from the proposition that $\sigma_{\mathscr{F}_{rad}}^{2}$ is the minimum attainable variance
within that class or equivalently 
\begin{equation}
\sigma_{\mathscr{F}_{rad}}^{2}\leq\sigma_{\mathscr{F}}^{2}
\end{equation}
Therefore, defining $\mathbb{T}_{rad} \equiv \left(N,\: \bar{P}_{i},\:\bar{\theta}_{j},\: \bar{T}_{i;j}\right)$ and using that $\Bigl\langle\mathscr{F}\Bigr\rangle=\Bigl\langle\mathscr{F}_{rad}\Bigr\rangle\leq0$,
we arrive at
\begin{equation}
\mathcal{G}[\mathbb{T}_{rad},M]\leq\mathcal{G}[\mathbb{T},M].
\end{equation}

\end{proof}

Thus, we have proven that whenever $W(x,p)$ is rotationally symmetric, a
rotationally invariant test is optimal. 
It should be noted, however, that the test described in the main text
is not as general as the one considered in this appendix
because the former is always a square polynomial to ensure non-negativity in a simple manner. Moreover,
the map we have used 
\begin{equation}
P_i\quad\rightarrow\quad\bar{P}_i
\end{equation}
does not preserve the property of being a square polynomial for $\mathscr{F}$, that is
$\mathscr{F}_{rad}=1+\sum_i \bar{P}_i$ is not generally a square polynomial even when $\mathscr{F}=1+\sum_i P_i$ is.
Nonetheless, the square root of $\mathscr{F}_{rad}$ can be
efficiently approximated by a high order
polynomial in $r^2$ in the relevant region where it is peaked at the negativity of the Wigner function. This, together with the convergence
shown in Fig.~\ref{fig: result graph}, motivates the choice of polynomials in $r^{2}$.

\subsection{Squeezed states.}
We will now consider the behaviour of the optimal test for the more general squeezed state. 
\begin{thm}(Invariance under squeezing)
Consider the best test $\mathbb{T}_{0}=\left(N,\: P_{i},\:\theta_{j},\: T_{i;j}\right) \in \mathscr{C}$
for the joint distribution $W(x,p)$. Then the best test	$\mathbb{T}_{\lambda}$ for the corresponding
squeezed state $W_{\lambda}(x,p)=W(\lambda x,p/\lambda)$, $\lambda$ being the squeezing parameter,
is characterised by $\mathbb{T}_{\lambda}=\left(N,\: P_{i}^{(\lambda)},\:\theta_{j}^{(\lambda)},\: T_{i;j}^{(\lambda)}\right)$
where 
\begin{equation}
P_{i}^{(\lambda)}(x,p)\equiv P_{i}(\lambda x,p/\lambda),
\label{eq: dilation of polynomials}
\end{equation}
\begin{equation}
\tan\theta_{j}^{(\lambda)}\equiv\frac{\tan\theta_{j}}{\lambda^{2}},
\label{eq: squeezing of cuts}
\end{equation}
\begin{equation}
T_{i;j}^{(\lambda)}\equiv T_{i;j}\left(\lambda\frac{\cos\theta_{j}}{\cos\theta_{j}^{(\lambda)}}\right)^{i}.
\label{eq: squeezed coefficients}
\end{equation}
\label{thm: A2}
\end{thm}
\begin{proof}
The proof is by construction of a bijective correspondence between tests for the unsqueezed case and tests for the squeezed one. First, we consider the representation
of the new set of polynomials in terms of the new set of cuts
\begin{eqnarray}
P_{i}^{(\lambda)}(x,p) & = & P_{i}(\lambda x,p/\lambda)\\
 & = & \sum_{j}T_{i;j}\left(\cos\theta_{j}\,\lambda x+\sin\theta_{j}\,\frac{p}{\lambda}\right)^{i} \nonumber \\
 & = & \sum_{j}T_{i;j}\left(\cos\theta_{j}\,\lambda\right)^{i}\left(x+\frac{\tan\theta_{j}}{\lambda^{2}}p\right)^{i}. \nonumber
\end{eqnarray}
In the domain $(-\pi,+\pi)$, we can find new cuts $\theta_{j}^{(\lambda)}$
according to Eq.~(\ref{eq: squeezing of cuts}) and obtain
\begin{eqnarray}
P_{i}^{(\lambda)}(x,p) & = & \sum_{j}T_{i;j}\left(\lambda\frac{\cos\theta_{j}}{\cos\theta_{j}^{(\lambda)}}\right)^{i}\left(\cos\theta_{j}^{(\lambda)}\, x+\sin\theta_{j}^{(\lambda)}\, p\right)^{i}
\nonumber \\
 & = & \sum_{j}T_{i;j}^{(\lambda)}Q_{\theta_{j}^{(\lambda)}}^{i},
\end{eqnarray}
where $T_{i;j}^{(\lambda)}$ is defined in Eq.~(\ref{eq: squeezed coefficients}).

We will now show that $\Bigl\langle\mathscr{F}\Bigr\rangle=\Bigl\langle\mathscr{F}^{(\lambda)}\Bigr\rangle_{\lambda}$
and that $\sigma_{\mathscr{F}}^{2}=\sigma_{\mathscr{F}^{(\lambda)}}^{2}$
where $\mathscr{F}$ ($\mathscr{F}^{(\lambda)}$) is the test function for the unsqueezed (squeezed) state. We recall the result for the Radon transform in Eq.~(\ref{eq: squeezing radon transform1}) stated below for convenience:
\begin{equation}
\mathbf{R}W_{\lambda}\;(\theta,s)=\frac{1}{u(\lambda,\theta)}\mathbf{R}W\;\left(\frac{s}{u(\lambda,\theta)}\right),
\label{eq: squeezing radon transform}
\end{equation}
where $\lambda$ is the squeezing parameter and $u(\lambda,\theta)$
is defined as $u(\lambda,\theta)=\frac{\cos\theta}{\lambda}\sqrt{1+\lambda^{4}\tan^{2}\theta}$.
It follows that 
\begin{eqnarray}
\Bigl\langle Q_{\theta}^{i}\Bigr\rangle_{\lambda} & \equiv & \int_{-\infty}^{+\infty}\textrm{d}s\, s^{i}\,\mathbf{R}W_{\lambda}\;(\theta,s)\nonumber \\
 & = & \int_{-\infty}^{+\infty}\frac{\textrm{d}s}{u(\lambda,\theta)}\, s^{i}\,\mathbf{R}W\;\left(\frac{s}{u(\lambda,\theta)}\right)\nonumber \\
 & = & u^{i}(\lambda,\theta)\int_{-\infty}^{+\infty}\textrm{d}t\, t^{i}\,\mathbf{R}W\;(\theta,t)\nonumber \\
 & \equiv & u^{i}(\lambda,\theta)\Bigl\langle Q^{i}\Bigr\rangle,
 \label{eq: quadrature moments under squeezing}
\end{eqnarray}
where we have changed from the variable $s$ to $t\equiv\frac{s}{u(\lambda,\theta)}$. 
For the mean values, we find  
\begin{eqnarray}
\Bigl\langle P_{i}^{(\lambda)}\Bigr\rangle_{\lambda} & \equiv & \sum_{j}T_{i;j}^{(\lambda)}\Bigl\langle Q_{\theta_{j}^{(\lambda)}}^{i}\Bigr\rangle_{\lambda}\\ \nonumber
 & = & \sum_{j}T_{i;j}\left(\lambda\frac{\cos\theta_{j}}{\cos\theta_{j}^{(\lambda)}}\right)^{i}u^{i}\left(\lambda,\theta_{j}^{(\lambda)}\right)\Bigl\langle Q^{i}\Bigr\rangle.
\end{eqnarray}
Since $\tan\theta_{j}^{(\lambda)}\equiv\frac{\tan\theta_{j}}{\lambda^{2}}$,
it is seen that $u^{i}\left(\lambda,\theta_{j}^{(\lambda)}\right)=\left(\lambda\frac{\cos\theta_{j}}{\cos\theta_{j}^{(\lambda)}}\right)^{-i}$ and
hence 
\begin{equation}
\Bigl\langle P_{i}^{(\lambda)}\Bigr\rangle_{\lambda}=\sum_{j}T_{i;j}\Bigl\langle Q^{i}\Bigr\rangle\equiv\Bigl\langle P_{i}\Bigr\rangle,
\end{equation}
so that $\Bigl\langle\mathscr{F}\Bigr\rangle=\Bigl\langle\mathscr{F}^{(\lambda)}\Bigr\rangle_{\lambda}$. 

We now wish to establish a similar result for the variances, $i.e.$, $\sigma_{\mathscr{F}^{(\lambda)}}^{2}=\sigma_{\mathscr{F}}^{2}$. We have that
\begin{equation}
\sigma_{\mathscr{F}^{(\lambda)}}^{2}=\frac{1}{2}\sum_{j}\sum_{i,i'=1}^{2N}g_{\lambda}^{ii'}(\theta_{j})\: T_{i;j}^{(\lambda)}T_{i';j}^{(\lambda)},
\end{equation}
where the covariance matrix $\frac{1}{2}g_{\lambda}^{ii'}(\theta_{j})\equiv\left(\left\langle Q_{\theta_{j}}^{i+i'}\right\rangle _{\lambda}-\left\langle Q_{\theta_{j}}^{i}\right\rangle_{\lambda}\left\langle Q_{\theta_{j}}^{i'}\right\rangle _{\lambda}\right)$ has been introduced.
Eq.~(\ref{eq: quadrature moments under squeezing}) implies that $g_{\lambda}^{ii'}(\theta_{j})=u^{i+i'}(\lambda,\theta_{j})g^{ii'}$
so that 
\begin{equation}
\sigma_{\mathscr{F}^{(\lambda)}}^{2}=\frac{1}{2}\sum_{j}\sum_{i,i'=1}^{2N}g^{ii'}T_{i;j}T_{i';j}\equiv\sigma_{\mathscr{F}}^{2},
\end{equation}
because $u^{i}\left(\lambda,\theta_{j}^{(\lambda)}\right)=\left(\lambda\frac{\cos\theta_{j}}{\cos\theta_{j}^{(\lambda)}}\right)^{-i}$.
Note that the variance for each cut $\sigma ^2 _j$ is preserved.

To conclude, we have explicitely constructed a correspondance between
any test used in the case of no squeezing $\mathbb{T}_{0}$ and a
test to use in the squeezed case $\mathbb{T}_{\lambda}$ such that
%\begin{equation}
%\frac{\Bigl\langle\mathscr{F}^{(\lambda)}\Bigr\rangle_{\lambda}}{\sigma_{\mathscr{F^{(\lambda)}}}}\Biggl|_{\mathbb{T}_{\lambda}}=\frac{\Bigl\langle\mathscr{F}\Bigr\rangle}{\sigma_{\mathscr{F}}}\Biggl|_{\mathbb{T}_{0}}
%\end{equation}

\begin{equation}
\mathcal{G}[\mathbb{T}_{\lambda},M]=\mathcal{G}[\mathbb{T},M].
\end{equation}

This result establishes that the best test for the squeezed state is at least as good as the one for the unsqueezed state. The transformation applied, however, also works in the opposite direction, so that a more optimal test for the squeezed state, could be unsqueezed and act as a more optimal test for the unsqueezed state. Hence there is a one to one correspondence between the optimal tests between the squeezed and unsqueezed states and the statement of Theorem~\ref{thm: A2} follows.
\end{proof}

\section{Inverse Radon method}
\label{sec: B}

In this appendix, we discuss the results of the filtered back-projection algorithm that have been quoted in the main
text without proof. We generalise and prove the result 
used for the linear transformed kernel (see Eq.~(\ref{eq: general kernel})), and provide the proof that the distribution of cuts described for the Monte Carlo quadrature method is the optimal distribution. Finally, we discuss the alternative procedure with a finite number of cuts in  detail.

The result for the linear transformation of the kernel can be derived for  the $d$-dimensional case. The planar case ($d=2$) that has been presented in the main text will then be retrieved as a special case. The whole theory of the Radon transform can be formulated within $\mathbb{R}^{d}$ where the Radon transform is achieved by integration over all hyperplanes in $\mathbb{R}^{d}$, rather than straight lines in the plane. For a detailed but not too abstract discussion see for instance Ref. \citep{Deans}. Here, we only focus on the adjoint Radon transformation. 
Let  $\mathbf{x}$ denote a vector in $\mathbb{R}^{d}$ 
and $\left\langle \,\cdot\,|\,\cdot\,\right\rangle$ denote the standard scalar product. 
A non singular linear transformation from $\mathbb{R}^{d}$ into itself is denoted by $L$ and
we introduce the $(d-1)$-dimensional sphere, $S^{d-1}$, as the set unit vectors in $\mathbb{R}^{d}$. The measure of the total set $S^{d-1}$ is
\begin{equation}
\left|S^{d-1}\right| \equiv \int_{\mathbf{n} \in S^{d-1}}d\mathbf{n}.
\end{equation}
We now give the explicit functional form of the adjoint Radon operator $\mathbf{R}^{\dagger}$, as it can be found in the reference mentioned above. This is an operator that maps kernel-like functions, call them
$\omega(\mathbf{n},s)$ where $\mathbf{n} \in S^{d-1}$ and $s \in \mathbb{R}$, into (filter-like) functions of $\mathbf{x} \in \mathbb{R}^{d}$, according to 
\begin{equation}
\mathbf{R}^{\dagger}\omega(\mathbf{x})=\frac{1}{\left|S^{d-1}\right|}\int_{\mathbf{n} \in S^{d-1}}\omega\left(\mathbf{n},\left\langle \mathbf{x}|\mathbf{n}\right\rangle \right)d\mathbf{n},
\end{equation}
Finally, we define $\omega_L$, the linear transformed kernel of $\omega$ through $L$, according to the identity
\begin{equation}
\mathbf{R}^{\dagger}\omega_{L}((\mathbf{x}))=\mathbf{R}^{\dagger}\omega(L(\mathbf{x})).
\label{eq: append lin transf kernel}
\end{equation}
If $\omega$ is the kernel for the filter $F$, $i.e.$,
$\mathbf{R}^{\dagger}\omega (\mathbf{x}) = F(\mathbf{x})$, then $\omega_{L}$ is the kernel for $F(L\mathbf{x})$. This
is precisely the situation encountered for the squeezed filter provided that $L$ is the squeezing transformation.
The following exact result holds.
%
%denoted by
%We consider the general case where the filter is defined on $\mathbb{R}^{m}$ even though only $m=2$ was considered in the main text. We define the vector of $m$ variables, $\mathbf{x}\in\mathbb{R}^{\mathrm{m}}$, and the invertible linear operator, $L$,  acting on the space $\mathbb{R}^{\mathrm{m}}$. We will later revisit the specific case of $\mathbb{R}^{2}$ in discussing the applications of the general results we derive.
%The Radon transform in $\mathbb{R}^{\mathrm{m}}$ is \cite{Deans}
%\begin{equation}
%\mathbf{R}^{\dagger}\omega(x)=\frac{1}{\left|S^{m-1}\right|}\int_{|\mathbf{n}|=1}\omega\left(\mathbf{n},\left\langle \mathbf{x}|\mathbf{n}\right\rangle \right)d\mathbf{n},
%\end{equation}
%where $\mathbf{n}\in$ $S^{m-1}$ is the collection of unit vectors in $\mathbb{R}^{\mathrm{m}}$,
%$\left|S^{m-1}\right|$ is the total measure of the $\left(m-1\right)$-sphere,
%$\left\langle \,\cdot\,|\,\cdot\,\right\rangle$ denotes the scalar
%product and the integration is carried out over $S^{m-1}$.
%Moreover, the space of all the hyperplanes in $\mathbb{R}^{\mathrm{m}}$
%is denoted by $Z^{\mathrm{m}}$. Each hyperplane is identified by
%the pair $(\mathbf{n},s)$, which is the direction and lenght of 
%the distance vector from the origin to the hyperplane, respectively. Formally, the adjoint operator acts like
%$\mathbf{R}^{\dagger}:\: \mathscr{S}\left(\mathbb{R}^{\mathrm{m}}\right)\:\rightarrow\:
%\mathscr{S}\left(Z^{\mathrm{m}}\right)$. Having settled the notation we can now state the first theorem.
\begin{thm}(Behaviour under linear transformation)
%Given the filter $\mathbf{R}^{\dagger}\omega(x)\in\mathscr{S}\left(\mathbb{R}^{\mathrm{m}}\right)$
%with kernel $\omega(\mathbf{n},s)\in\mathscr{S}\left(Z^{\mathrm{m}}\right)$
%and an invertible linear transformation $L$, the kernel $\omega_{L}(\mathbf{n},s)$,
%defined in such a way that $\mathbf{R}^{\dagger}\omega_{L}(x)=\mathbf{R}^{\dagger}\omega(Lx)$,
Given the filter $\omega(\mathbf{n},s)$ and an invertible linear transformation $L$, then the linear transformed kernel $\omega_{L}(\mathbf{n}',s)$ is given by 
\begin{eqnarray}
\omega_{L}(\mathbf{n}',s) & = & \left|\det\mathbb{J}_{\phi}(\mathbf{n})\right|{}^{-1}\omega(\mathbf{n},\left|L^{T}\mathbf{n}\right|s).
\end{eqnarray}
We have here put $\mathbf{n}'=\phi(\mathbf{n})\equiv\frac{L^{T}\mathbf{n}}{\left|L^{T}\mathbf{n}\right|}$,
which is invertible on the semisphere, and $\mathbb{J}_{\phi}$ is its Jacobian matrix. Note that the substitution $\mathbf{n}=\phi^{-1}(\mathbf{n}')$
is intended in the formula above.
\label{thm: B1}
\end{thm}
\begin{proof}
The application of the adjoint operator to the identity that defines
$\omega_{L}(x)$ (\ref{eq: append lin transf kernel}) leads to
\begin{equation}
\mathbf{R}^{\dagger}\omega(L\mathbf{x})=\frac{1}{\abs{S^{d-1}}}\int_{\mathbf{n} \in S^{d-1}}\omega(\mathbf{n},\left\langle L\mathbf{x}|\mathbf{n}\right\rangle )d\mathbf{n}.
\end{equation}
Considering the identity $\dibraket{L\mathbf{x}}{\mathbf{n}} =\dibraket{\mathbf{x}}{L^{T}\mathbf{n}}$
with $L^{T}$ denoting the transposed transformation, it is natural to change
the variable of the integration to $\mathbf{n}'=\phi(\mathbf{n})\equiv\frac{L^{T}\mathbf{n}}{|L^{T}\mathbf{n}|}$,
$\mathbf{n}'\in S^{d-1}$, which leads to
\begin{multline}
\mathbf{R}^{\dagger}\omega_{L}(\mathbf{x})=\mathbf{R}^{\dagger}\omega(L\mathbf{x})\\
 =\frac{1}{\left|S^{d-1}\right|}\int_{\mathbf{n}' \in S^{d-1}}\omega(\mathbf{n},\abs{L^{T}\mathbf{n}} \dibraket{\mathbf{x}}{\mathbf{n}'})\frac{d\mathbf{n}'}{\left|\det\mathbb{J}_{\phi}(\mathbf{n})\right|},
\label{eq: integral}
\end{multline}
which is the adjoint Radon transform of the function $\omega_{L}(\mathbf{n}',s)$. The statement of Theorem~\ref{thm: B1} directly follows by putting $\dibraket{\mathbf{x}}{\mathbf{n}'} = s$.
\end{proof}

For our application of the theorem, we consider a phase plane, $\mathbb{R}^{\mathrm{2}}$,
and the linear transformation $L$ corresponds to a squeezing transformation (the case of dilation, $i.e.$ $L$ is a multiple of the identity is trivial).
We write $\mathbf{n}(\theta)=(\cos\theta,\sin\theta)$
and recover the notation used for the kernel in the main text, $i.e.$,
$\omega(\theta,s)\equiv\omega(\mathbf{n}(\theta),s)$. 
%If we parametraize
%the general dilation by $L(a)=\left(\begin{array}{cc}
%a & 0\\
%0 & a
%\end{array}\right)$, it is seen that $\left|L^{T}\mathbf{n}\right|=a$ and $\phi(\mathbf{n})\equiv\frac{L^{T}\mathbf{n}}{|L^{T}\mathbf{n}|}=\mathbf{n}$, so that to the trivial result $\omega_{L(a)}(s)=\omega(s/a)$ ensues.
%The dilation, however, does not preserve the Lebesgue measure of the plane and the output of the transformation should therefore be normalized as~\cite{Niev86}
%\begin{equation}
%\omega_{a}(s)=\frac{1}{a^{2}}\omega\left(\frac{s}{a}\right).
%\label{eq: rescaled kernel}
%\end{equation}
%Note the simplification in notation
%$\omega_{a}(s)\equiv\omega_{L(a)}(s)$.
The following corollary holds.

\begin{crl}
The general filter that can be obtained from a rotational invariant
one, denoted by $\omega(s)$, by dilation and rescaling is of the form
\begin{equation}
\omega_{a,\lambda}(\theta,s)=\frac{1}{u^{2}(\lambda,\theta)}\,\omega_{a}\left(\frac{s}{u(\lambda,\theta)}\right),
\label{eq: general kernel-1}
\end{equation}
with $\omega_{a}(s)$ defined in Eq.~(\ref{eq: dilated approximate cylindrical dirac measure}) so that
$\omega_{a,\lambda}(\theta,s)$ is normalized as well. $\lambda$
is the rescaling parameter. 
\end{crl}
\begin{proof}
We parametrize the rescaling transformation by $L(\lambda)=\left(\begin{array}{cc}
\lambda & 0\\
0 & \lambda^{-1}
\end{array}\right)$. We also introduce the the following parametrizations $\mathbf{n}(\theta)=(\cos\theta,\sin\theta)$ and
$\mathbf{n}'(\theta')=(\cos\theta',\sin\theta')$. Then one can write
\begin{eqnarray}
\tan\theta' & = & \lambda^{-2}\tan\theta\nonumber \\
\left|L^{T}\mathbf{n}\right| & = & u^{-1}(\theta',\lambda) \nonumber\\
\left|\det\mathbb{J}_{\phi}(\mathbf{n})\right|{}^{-1} & = & \frac{1+\tan^{2}\theta'}{\lambda^{-2}+\tan^{2}\theta'}=\left|L^{T}\mathbf{n}\right|{}^{2} 
\end{eqnarray}
and by substitution the proof is concluded. Note the simplification in notation
$\omega_{a,\lambda}(\theta,s)\equiv\omega_{a,L(\lambda)}(\theta,s)$.
\end{proof}

At this point, we have established the general results on the filter back-projection that we have used throughout the application of the Inverse Radon method. We now turn to the particular problem of the determinating the best distribution of cuts $\mathcal{C}(\theta)$,
(see Eq.~(\ref{eq: Montecarlo splitting})),
for the Monte Carlo quadrature procedure.

\subsection{Monte Carlo quadrature.}
By definition, the mean value of the test function and $\Delta_a(\epsilon)$ do not depend on $\mathcal{C}(\theta)$ nor on $\lambda$ [see Eqs.~(\ref{eq: justifying the squeezed quadrature}) and (\ref{eq: systematic error})]. Nonetheless, the variance of the estimate $\langle F_{a,\lambda}\rangle$ has a strong dependence on $\mathcal{C}(\theta)$, that is
\begin{eqnarray}
&\sigma^2_{F_{a,\lambda}} [\mathcal{C};\epsilon,M] \equiv \quad\quad \quad\quad\quad\quad\quad\quad\quad\quad\quad\quad\quad\quad\quad \nonumber \\ 
 &\equiv \int_{-\frac{\pi}{2}}^{+\frac{\pi}{2}}\int_{\mathbb{R}}
p_\lambda (\theta,s) \frac{\left(I_{a,\lambda,\epsilon}(\theta,s)  - \langle F_{a,\lambda} \rangle \right)^2}{M-1} dsd\theta \nonumber \\
& \equiv \frac{\mathcal{W}_{a,\epsilon}}{M-1} \Big(\int_{-\frac{\pi}{2}}^{+\frac{\pi}{2}}  \frac{d\theta}{\mathcal{C}(\theta)u^4(\lambda,\theta)} \Big)- \frac{\langle F_{a,\lambda} \rangle^2}{M-1}, \quad\quad
\end{eqnarray}
where the quantity $\mathcal{W}_{a,\epsilon} \equiv \left(\int_{\mathbb{R}}\mathbf{R}W (t)\frac{\omega_{a,\epsilon}^2(t)}{\pi^2} dt\right)$ has been defined. Note that this quantity only relates to the unsqueezed distribution and all the $\lambda$-dependence of the variance $\sigma^2_{F_{a,\lambda}}$ is contained in the integral $\Big(\int_{-\frac{\pi}{2}}^{+\frac{\pi}{2}}  \frac{d\theta}{\mathcal{C}(\theta)u^4(\lambda,\theta)} \Big)$, which also depends upon the distribution $\mathcal{C}$. Since our present concern is a suitable choice of $\mathcal{C}$, the
integral in the expression above is the main focus of the  analysis. 
In other words, for fixed $a$, $\lambda$, $\epsilon$ and $M$, we wish to determine the distribution $\mathcal{C}$ that minimises  $\Big(\int_{-\frac{\pi}{2}}^{+\frac{\pi}{2}}  \frac{d\theta}{\mathcal{C}(\theta)u^4(\lambda,\theta)} \Big)$ and $\sigma^2_{F_{a,\lambda}} [\mathcal{C};\epsilon,M]$ accordingly. We can now state the following exact result concerning the relevant ratio in Eq.~(\ref{eq: ratio IR}).

\begin{crl}
Assume the state to be described by a squeezed distribution $W_\lambda (x,p)$, where for $\lambda = 1$ the distribution is rotationally invariant. Then the optimal distribution of cuts $\mathcal{C}_o$, which minimises the relevant ratio~(\ref{eq: ratio IR}), only depends  on $\lambda$ through
\begin{equation}
\mathcal{C}_o(\theta) \equiv \frac{1}{\pi u^2 (\lambda,\theta)}.
\label{eq: Montecarlo best C}
\end{equation}
Moreover, for this choice of phase cut distribution the value of the relevant ratio does not depend on $\lambda$ and it is thus equal to the one for the rotationally invariant state $(\lambda=1)$.
\label{crl: B2}
\end{crl}
\begin{proof}

We have already pointed out that the dependence on the distribution of cuts in the ratio $\frac{\langle F_{a,\lambda} \rangle}{\sigma_{F_{a,\lambda}}+\Delta_a (\epsilon)}$  is only through the integral $\Big(\int_{-\frac{\pi}{2}}^{+\frac{\pi}{2}}  \frac{d\theta}{\mathcal{C}(\theta)u^4(\lambda,\theta)} \Big)$ appearing in the variance. We are thus 
considering a constrained optimisation problem and can define the Lagrangian functional
\begin{equation}
\mathscr{L}[\mathcal{C}, \mu ] \equiv \int_{-\frac{\pi}{2}}^{+\frac{\pi}{2}}  \Bigg( \frac{1}{\mathcal{C}(\theta)u^4(\lambda,\theta)} + \mu  \mathcal{C}(\theta) \Bigg) d\theta .
\end{equation}
together with the normalisation condition, $\int_{-\frac{\pi}{2}}^{+\frac{\pi}{2}} \mathcal{C}(\theta) d\theta = 1$, that
determines the Lagrange multiplier $\mu \in \mathbb{R}$.
Let us consider a small variation of the distribution of cuts, $\mathcal{C}(\theta)\rightarrow\mathcal{C}(\theta)+\delta\mathcal{C}(\theta)$. Then, the stationary point equation is given by 
\begin{equation}
0 = \delta\mathscr{L}[\mathcal{C}, \mu ] = 
\int_{-\frac{\pi}{2}}^{+\frac{\pi}{2}} \Big( -\frac{1}{\mathcal{C}^2(\theta)u^4(\lambda,\theta)} + \mu \Big)\delta\mathcal{C}(\theta)d\theta,
\label{eq: append radon stationary C eq}
\end{equation}
where $\delta\mathscr{L}[\mathcal{C}, \mu ]$ is the linear part of $\mathscr{L}[\mathcal{C}+\delta\mathcal{C}, \mu ]-\mathscr{L}[\mathcal{C}, \mu ]$ with respect to $\delta\mathcal{C}$.
The solution of Eq.~(\ref{eq: append radon stationary C eq}) is the stationary distribution of cuts $\mathcal{C}_o = \frac{1}{\sqrt{\mu} \,u^2(\lambda,\theta)}$. By integration, 
\begin{align}
\int_{-\pi/2}^{+\pi/2}\frac{d\theta}{u^{2}(\lambda,\theta)} & =\int_{-\pi/2}^{+\pi/2}\frac{\lambda^{2}}{\cos^{2}\alpha}\frac{d\alpha}{1+\lambda^{4}\tan^{2}\alpha}\nonumber \\
 & =\int_{-\infty}^{+\infty}\frac{dx}{1+x^{2}}=\pi,
 \label{eq: Montecarlo int u^-2}
\end{align} 
we find the normalization condition $\mu = \pi^2$. 
We can now check that the stationary distribution $\mathcal{C}_o$ gives the minimum of the Lagrangian, by considering the second variation around $\mathcal{C}_o$, that is
\begin{eqnarray}
\delta^2\mathscr{L}[\mathcal{C}_o, \mu ] = \pi^3\mathcal{W}_a \int_{-\frac{\pi}{2}}^{+\frac{\pi}{2}}\delta\mathcal{C}^2(\theta)u^2(\lambda,\theta) d\theta > 0 , \nonumber \\
\label{eq: variation IR}
\end{eqnarray}
which is positive for any variation $\delta\mathcal{C}$. $\mathcal{C}_o$ thus gives the minimum. Finally, inserting $\mathcal{C}_o$ into $\sigma^2_{F_{a,\lambda}} [\mathcal{C};\epsilon,M] $  shows that the resulting variance $\sigma^2_{F_{a,\lambda}} [\mathcal{C};\epsilon,M] $ is independent of $\lambda$
This completes the proof of Colloray~\ref{crl: B2}.

\end{proof}
The naive choice of uniformly distributed cuts $\mathcal{C} = \frac{1}{\pi}$ leads to a poor variance  for a general squeezed state.  This is easily seen by straightforward computation of the integral
\begin{align}
\int_{-\pi/2}^{+\pi/2}\frac{d\theta}{u^{4}(\theta,\lambda)} & =\int_{-\pi/2}^{+\pi/2}\frac{\lambda^{4}}{\cos^{4}\alpha}\frac{d\alpha}{\left(1+\lambda^{4}\tan^{2}\alpha\right)^{2}}\\
 & =\int_{-\infty}^{+\infty}\frac{\lambda^{2}+\lambda^{-2}x^{2}}{\left(1+x^{2}\right)^{2}}dx=\frac{\pi}{2}\left(\lambda^{2}+\lambda^{-2}\right),
\end{align}
so that 
\begin{equation}
\sigma^2_{F_{a,\lambda}} = \frac{\mathcal{W}_{a,\epsilon}}{M-1} \frac{\left(\lambda^{2}+\lambda^{-2}\right)}{2} - \frac{\langle F_{a,\lambda} \rangle^2}{M-1}.
\end{equation} 
The first term is dominant for $\lambda \gg 1$ ($\lambda \ll 1$), which gives the scaling of $\sim \lambda^2$ ($\sim \lambda^{-2}$).

\subsection{Finite number of cuts.}

Here, we consider the details of the inverse Radon method with a finite number of cuts in phase space. As remarked in the main text, we assume that the states are described by a symmetric $W_\lambda(x,p)$ phase space distribution, where for $\lambda=1$ this is the rotationally symmetric single photon state.
We wish to classify the best way to carry out the approximation  
\begin{eqnarray}
 \langle F_{a,\lambda} \rangle & \simeq & \sum_{\{\theta_j\}}w_j g_{a,\lambda}(\theta_j)
\label{eq: final quadrature}
\end{eqnarray}
with only a limited number of phase cuts. 
By the symmetry assumed, the mean value of the test
function is (see Eqs.~(\ref{eq: general kernel},\ref{eq: squeezing radon transform1},\ref{eq: justifying the squeezed quadrature}))
\begin{eqnarray}
\langle F_{a,\lambda} \rangle & = & \frac{1}{\pi}\int_{-\frac{\pi}{2}}^{+\frac{\pi}{2}}\frac{d\theta}{ u^2(\lambda,\theta)} \Bigg( \int_{\mathbb{R}}\mathbf{R}W_{\lambda=1}\;(t)\,\omega_{a,\lambda=1}(t)\, dt \Bigg)	\nonumber \\
 &=	& \frac{1}{\pi}\int_{-\frac{\pi}{2}}^{+\frac{\pi}{2}} d\theta \ \frac{1}{ u^2(\lambda,\theta)} g_{a, \lambda=1}.
\end{eqnarray}

Here, it is important that $g_{a, \lambda=1}$ does not depend on the phase. The problem thus concerns how to carry out the integration of $\frac{1}{u^2(\theta,\lambda)}$ numerically.
In other words, if $g_{a, \lambda=1}$ does not contain any experimental imprecision, the only error associated to the estimate $\langle F_{a,\lambda} \rangle$ comes from the discretisation of $\int\frac{d\theta}{u^2(\theta,\lambda)}$.
\begin{figure}[t!]
\centering
\includegraphics[scale=.4, trim= 25mm 20mm 25mm 20mm]{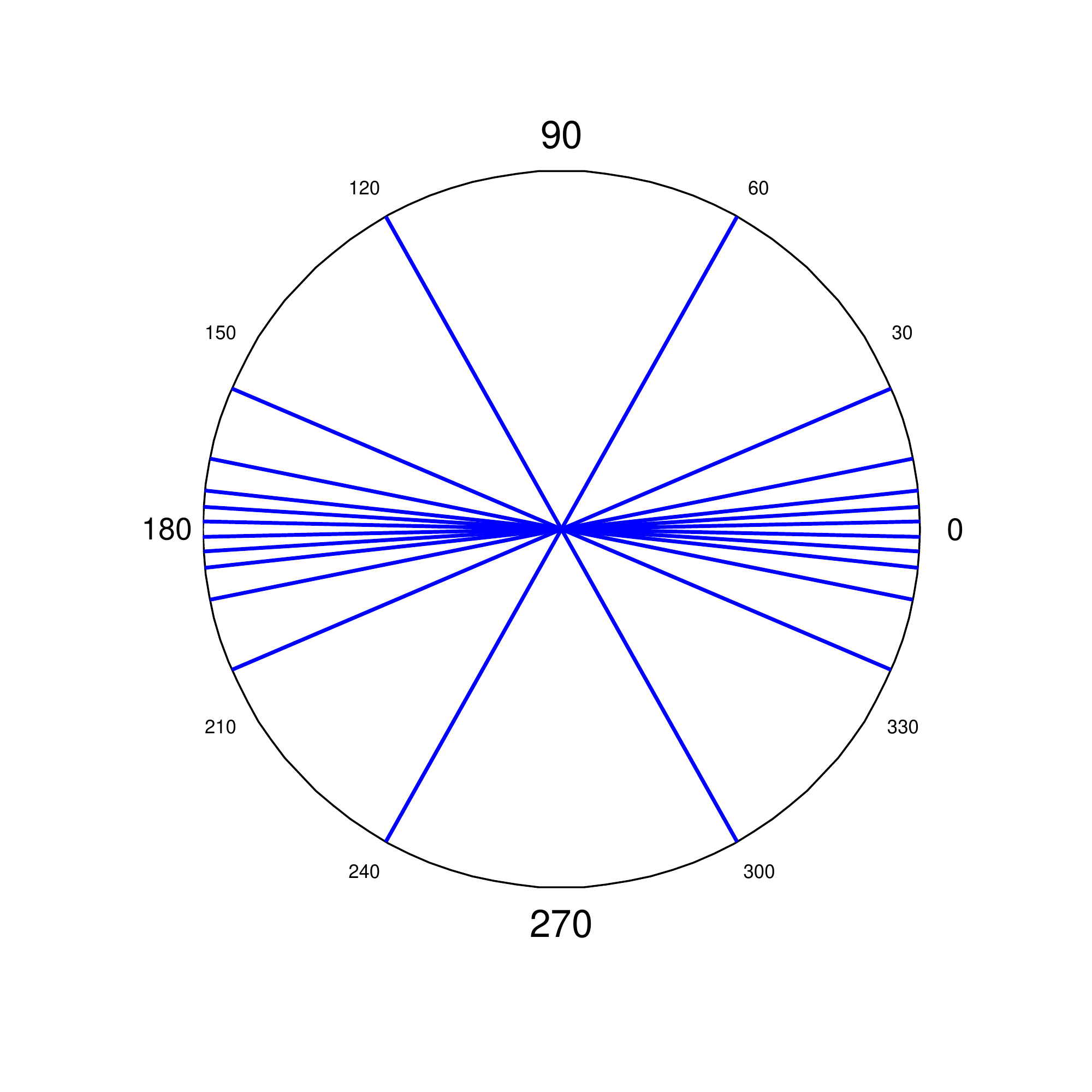}
\caption{Optimal postion for $12$ cuts in the middle-point quadrature. Here $\lambda=5$
and the precision of the numerical quadrature is within the $3 \%$ of error.}
\label{fig: position of cuts numerical quadrature}
\end{figure} 
This problem can be addressed with use of the composite middle-point formula. 
%First we note that since the integrand
%is even, the problem is identically solved in either interval $[0,+\frac{\pi}{2}]$ or $[-\frac{\pi}{2},0]$. We therefore only consider the former interval.
The numerical quadrature is obtained by choosing a 
partition of the interval $\left[ -\frac{\pi}{2},\frac{\pi}{2} \right]$ into non-uniform subintervals, in which 
the value of the integral is approximated by the product of the length of the subinterval times the value of the integrand in its middle point. We exploit the parity of the integrand and choose partition points symmetric around $\theta = 0$, the latter included. 
The complete integral is then approximated as
\begin{equation}
\frac{1}{\pi}\int_{-\frac{\pi}{2}}^{+\frac{\pi}{2}}d\theta \ \frac{1}{ u^2(\lambda,\theta)} \simeq \frac{1}{\pi} \sum_{j=1}^m  \frac{\Delta \theta _j}{u^2(\tilde{\theta}_j,\lambda)},
\label{eq: append middle-point formula}
\end{equation}
where $\tilde{\theta}_j$ and $\Delta \theta _j$ are  
the middle point and the size of the $j$-th subinterval, respectively. The middle points should
correspond to the cuts where measurements are carried out. In analogy with the elementary test, $m$ is the total number of cuts over the entire interval 
$\left[ -\frac{\pi}{2},\frac{\pi}{2} \right]$. By comparison with Eq.~(\ref{eq: final quadrature}), 
the weights are $w_j = \frac{\Delta \theta _j}{\pi}$.
The evaluation at the middle point is a convenient choice because
the integrand is monotonically decreasing (increasing) within the interval $[0,+\frac{\pi}{2}]$ ($[-\frac{\pi}{2},0]$). The middle-point evaluation thus
provides a trade off between the surplus in one half of the interval
and the deficiency in the other.

We define the numerical error to be
\begin{equation}
\mathscr{E} \equiv \Bigg|1 - \frac{1}{\pi} \sum_{j=1}^m \frac{\Delta \theta _j}{u^2(\tilde{\theta}_j,\lambda)} \Bigg|.
\end{equation}
For fixed $\lambda$ and fixed $m$, we have numerically optimised $\mathscr{E}$
with respect to the position of the cuts. As an example, we show in Fig.~\ref{fig: position of cuts numerical quadrature} the result for $\lambda = 5$ and $m = 12$.
We observe a concentration of cuts near $\theta = 0$, in analogy with
the elementary test. 
Moreover, as shown in Fig.~\ref{fig: error quadrature}, $\mathscr{E}$ decreases roughly exponentially as a function of $m$. The error can, in principle be made arbitrarily small by
increasing the number of cuts. The bigger $\lambda$ is, however, the more cuts are needed to push the error below a given threshold. Hence, in an actual experiment with a squeezed state, it will be necessary to use a large number of cuts to constrain the error from the numerical quadrature.  Also, one would have to justify using a highly asymmetric distribution of cuts to estimate an integral over all angles. This can be argued from the asymmetry of the filter, which is the main reason for strong angular dependence, but such an argument would still leave open the question of whether the state itself have rapid angular variations in the region with a less dense distribution of cuts. Alternatively, one would have to resort to very small intervals between the cuts for all angles. If one is using the same number of measurements per cut, this would be very costly in terms of measurements since a large number of measurements are distributed in a region with little dependence on the angle. One can to some degree  compensate for this  with an asymmetric distribution of the number of measurement (see below), but still the issue of the precision of the numerical quadrature is an uncomfortable question, which it would be desirable to avoid, $e.g.$, by using the elementary test. 

Note also that, for simplicity, the discussion  above has been
carried out by assuming exact knowledge of $g_{a,\lambda}$. In an actual experiment, where we want to present a proof of non-classicality with as few assumptions as possible, we do not want to rely on a specific theoretical model for the underlying distribution. Instead, one would have to justify a specific bound on the error coming from having a finite number of cuts by considering the actual measured experimental data. This is significantly more complicated than the analysis performed here, and we shall not discuss it.

\begin{figure}[t!]
\centering
\includegraphics[scale=0.43]{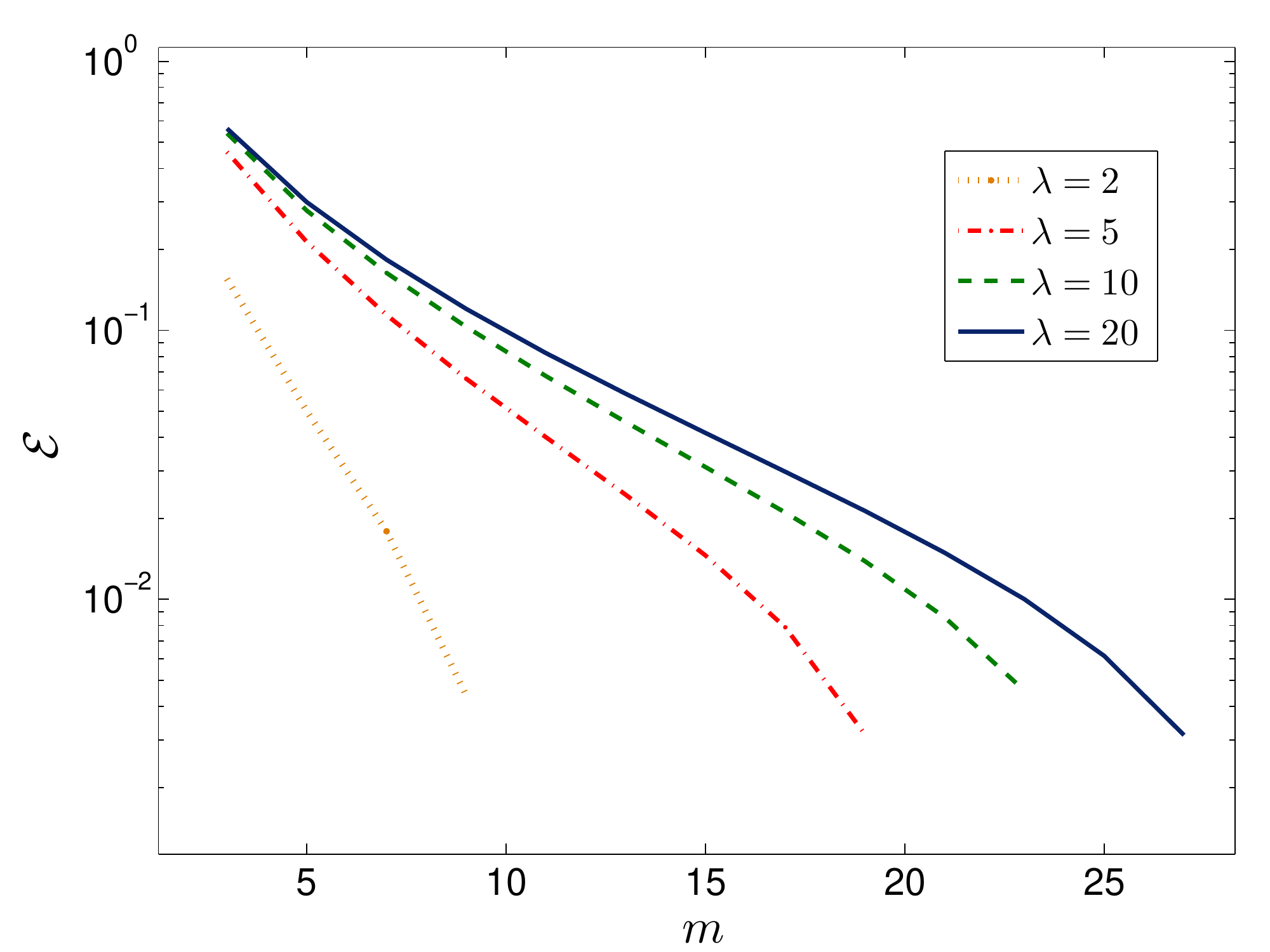}
\caption{The numerical error $\mathscr{E}$ as a function of the number of cuts $m$ for various values of the squeezing parameter $\lambda$. The plot has been truncated at the onset of numerical instabilities that causes the error to oscillate around $10^{-5}$. 
Notice the nearly exponential decrease for each curve and that for larger $\lambda$ more cuts are needed to reach a given error.}	
\label{fig: error quadrature}
\end{figure}

In addition to the error from the finite number of cuts, we also we have to consider the statistical error.
If  the cuts are considered to be statistically independent, the variance of the estimate $\langle F_{a,\lambda} \rangle$ 
is 
\begin{equation}
\sigma ^2 _{F_{a,\lambda}} = \sum_{ \{ \theta_j \} } \frac{W^2_j}{M_j - 1} \sigma ^2 _{a},
\label{eq: variance of the quadrature}
\end{equation}
where for brevity we have put $W_j = \frac{w_j}{u^2(\tilde{\theta}_j,\lambda)}$, satisfying $\sum_j W_j = 1$. Note that the $\lambda$ dependence is  implicit in
the optimal $\{ \theta_j \}$ and $W_j$. Also, the variance of $g_{a}(\theta_j)$ is independent of $\theta_j$.
For uniform weights, it is desirable to have the same number of measurements $M_j$ for each cut.
For the optimal quadrature, where $W_j$ are not uniform, this is no longer true, since a small weight $W_j$ allows to tolerate a rather big imprecision in $g_{a}(\theta_j)$. 
Given the total number of measurement $M$ and $m$ cuts, the optimal
distribution of measurements per cut $M_j$ can be obtained analytically.
Formally, the problem consists in the
optimisation of the quantity in Eq.~(\ref{eq: variance of the quadrature})
with respect to $M_j$, constrained by $\sum_j (M_j -1) = M - m$.
The Lagrangian for this optimisation problem is 
\begin{equation}
\mathscr{L}({M_j},\mu) = \sigma ^2 _{a}\sum_{j=1}^m \frac{W^2_j}{M_j-1} +\mu \sum_{ j=1 }^m (M_j-1),
\label{eq: lagrangian radon transform}
\end{equation}
which has the optimal solution
\begin{equation}
\begin{cases}
M_j - 1 = \frac{\sigma _a W_j}{\sqrt{\mu}} \\
\mu = \left( \frac{\sigma _a \sum _j W_j}{M-m}\right)^2
\end{cases}
\end{equation}
By eliminating the Lagrange multiplier $\mu$,  we find $(M_j -1) = \frac{(M - m) W_j}{\sum _j W_j}$
and the minimal value of the variance is 
\begin{equation}
\sigma ^2 _{F_{a,\lambda}} =  \frac{\left( \sum_{j=1}^m W_j \right)^2}{M-m} \sigma ^2 _{a} = \frac{\sigma ^2 _{a}}{M-m}
\label{eq: optimal variance}.
\end{equation}
The last equality follows from $\sum_{\{ \theta_j \}}W_j = 1$ and establishes the independence on the squeezing, similar to the results  obtained for the other methods.
Notice that we have considered $M_j \in \mathbb{R}$ for simplicity; the fact
that they are natural numbers introduce a small deviation from the optimum determined here.

The total error associated with the estimate $\langle F_{a,\lambda} \rangle$ is given by the sum of all contributions
\begin{equation}
\sigma ^2 _{F_{a,\lambda}} + \mathscr{E} + \Delta_a(\epsilon).
\end{equation}
Similar to the other test this can be made independent of $\lambda$ provided that $m$ is big enough (see also Fig.~\ref{fig: error quadrature} and discussion; notice although that a large $ m$ can decrease the strength of the test if $M$ is not very big). Due to the large number of ambiguities which would have to be resorted in order to make a sensible comparison with the elementary test, we will, however,  not make a more detailed comparison.

%\bibliographystyle{plain}
%\bibliography{article}
%\bibliography{PAPER_WORKPROGRESS_AUTUMN}

\begin{thebibliography}{10}%
\makeatletter
\providecommand \@ifxundefined [1]{%
 \ifx #1\undefined \expandafter \@firstoftwo
 \else \expandafter \@secondoftwo
\fi
}%
\providecommand \@ifnum [1]{%
 \ifnum #1\expandafter \@firstoftwo
 \else \expandafter \@secondoftwo
\fi
}%
\providecommand \enquote [1]{``#1''}%
\providecommand \bibnamefont  [1]{#1}%
\providecommand \bibfnamefont [1]{#1}%
\providecommand \citenamefont [1]{#1}%
\providecommand\href[0]{\@sanitize\@href}%
\providecommand\@href[1]{\endgroup\@@startlink{#1}\endgroup\@@href}%
\providecommand\@@href[1]{#1\@@endlink}%
\providecommand \@sanitize [0]{\begingroup\catcode`\&12\catcode`\#12\relax}%
\@ifxundefined \pdfoutput {\@firstoftwo}{%
 \@ifnum{\z@=\pdfoutput}{\@firstoftwo}{\@secondoftwo}%
}{%
 \providecommand\@@startlink[1]{\leavevmode}%
 \providecommand\@@endlink[0]{}%
}{%
 \providecommand\@@startlink[1]{%
  \leavevmode
  \pdfstartlink
   attr{/Border[0 0 1 ]/H/I/C[0 1 1]}%
   user{/Subtype/Link/A<</Type/Action/S/URI/URI(#1)>>}%
  \relax
 }%
 \providecommand\@@endlink[0]{\pdfendlink}%
}%
\providecommand \url  [0]{\begingroup\@sanitize \@url }%
\providecommand \@url [1]{\endgroup\@href {#1}{\urlprefix}}%
\providecommand \urlprefix [0]{URL }%
\providecommand \Eprint[0]{\href }%
\@ifxundefined \urlstyle {%
  \providecommand \doi [1]{doi:\discretionary{}{}{}#1}%
}{%
  \providecommand \doi [0]{doi:\discretionary{}{}{}\begingroup
  \urlstyle{rm}\Url }%
}%
\providecommand \doibase [0]{http://dx.doi.org/}%
\providecommand \Doi[1]{\href{\doibase#1}}%
\providecommand \bibAnnote [3]{%
  \BibitemShut{#1}%
  \begin{quotation}\noindent
    \textsc{Key:}\ #2\\\textsc{Annotation:}\ #3%
  \end{quotation}%
}%
\providecommand \bibAnnoteFile [2]{%
  \IfFileExists{#2}{\bibAnnote {#1} {#2} {\input{#2}}}{}%
}%
\providecommand \typeout [0]{\immediate \write \m@ne }%
\providecommand \selectlanguage [0]{\@gobble}%
\providecommand \bibinfo [0]{\@secondoftwo}%
\providecommand \bibfield [0]{\@secondoftwo}%
\providecommand \translation [1]{[#1]}%
\providecommand \BibitemOpen[0]{}%
\providecommand \bibitemStop [0]{}%
\providecommand \bibitemNoStop [0]{.\EOS\space}%
\providecommand \EOS [0]{\spacefactor3000\relax}%
\providecommand \BibitemShut [1]{\csname bibitem#1\endcsname}%
%</preamble>
\bibitem{Dirac26}%
  \BibitemOpen
  \bibfield{author}{%
  \bibinfo {author} {\bibfnamefont{P.}~\bibnamefont{Dirac}},\ }%
  \emph{\bibinfo {title} {Principles of Quantum Mechanics}}\ (\bibinfo
  {publisher} {SIAM Society for Industrial and Applied Mathematics},\ \bibinfo
  {year} {1926})%
  \bibAnnoteFile{NoStop}{Dirac26}%
\bibitem{vonNeumann}%
  \BibitemOpen
  \bibfield{author}{%
  \bibinfo {author} {\bibfnamefont{J.}~\bibnamefont{von Neumann}},\ }%
  \emph{\bibinfo {title} {Mathematische Grundlagen der Quantenmechanik}}\
  (\bibinfo {year} {1932})%
  \bibAnnoteFile{NoStop}{vonNeumann}%
\bibitem{thooft14}%
  \BibitemOpen
  \bibfield{author}{%
  \bibinfo {author} {\bibnamefont{'t~Hooft~Gerard}},\ }%
  \bibfield{journal}{%
  \bibinfo {journal} {arXiv preprint arXiv:14051548v2.0885}}%
   (\bibinfo {month} {Jun}\ \bibinfo {year} {2014}),\
  \Eprint{http://arxiv.org/abs/14051548v2.0885}{14051548v2.0885 [quant-ph]}%
  \bibAnnoteFile{NoStop}{thooft14}%
\bibitem{Frohlich12}%
  \BibitemOpen
  \bibfield{author}{%
  \bibinfo {author} {\bibfnamefont{J.}~\bibnamefont{Fr\"ohlich}}\ and\ \bibinfo
  {author} {\bibfnamefont{B.}~\bibnamefont{Schubnel}},\ }%
  \bibfield{journal}{%
  \bibinfo {journal} {arXiv preprint arXiv:1203.3678}}%
   (\bibinfo {year} {2012})%
  \bibAnnoteFile{NoStop}{Frohlich12}%
\bibitem{NielsenChuang}%
  \BibitemOpen
  \bibfield{author}{%
  \bibinfo {author} {\bibfnamefont{M.~A.}\ \bibnamefont{Nielsen}}\ and\
  \bibinfo {author} {\bibfnamefont{I.~L.}\ \bibnamefont{Chuang}},\ }%
  \emph{\bibinfo {title} {Quantum Computation and Quantum Information
  (Cambridge Series on Information and the Natural Sciences)}}\ (\bibinfo
  {publisher} {Cambridge University Press, Cambridge, 2004},\ \bibinfo {year}
  {2004})%
  \bibAnnoteFile{NoStop}{NielsenChuang}%
\bibitem{Kot12}%
  \BibitemOpen
  \bibfield{author}{%
  \bibinfo {author} {\bibfnamefont{E.}~\bibnamefont{Kot}}, \bibinfo {author}
  {\bibfnamefont{N.}~\bibnamefont{Gr\o{}nbech-Jensen}}, \bibinfo {author}
  {\bibfnamefont{B.~M.}\ \bibnamefont{Nielsen}}, \bibinfo {author}
  {\bibfnamefont{J.~S.}\ \bibnamefont{Neergaard-Nielsen}}, \bibinfo {author}
  {\bibfnamefont{E.~S.}\ \bibnamefont{Polzik}},\ and\ \bibinfo {author}
  {\bibfnamefont{A.~S.}\ \bibnamefont{S\o{}rensen}},\ }%
  \bibfield{journal}{%
  \Doi{10.1103/PhysRevLett.108.233601}{\bibinfo {journal} {Phys. Rev. Lett.}}\
  }%
  \textbf{\bibinfo {volume} {108}},\ \bibinfo {pages} {233601} (\bibinfo
  {month} {Jun}\ \bibinfo {year} {2012}),\
  \url{http://link.aps.org/doi/10.1103/PhysRevLett.108.233601}%
  \bibAnnoteFile{NoStop}{Kot12}%
\bibitem{Holevo11}%
  \BibitemOpen
  \bibfield{author}{%
  \bibinfo {author} {\bibfnamefont{A.~S.}\ \bibnamefont{Holevo}},\ }%
  \emph{\bibinfo {title} {Probabilistic and Statistical Aspects of Quantum
  Theory}},\ \bibinfo {edition} {1st}\ ed.\ (\bibinfo {publisher} {Pisa:
  Edizioni della Normale},\ \bibinfo {year} {2011})%
  \bibAnnoteFile{NoStop}{Holevo11}%
\bibitem{Cavalcanti14}%
  \BibitemOpen
  \bibfield{author}{%
  \bibinfo {author} {\bibfnamefont{N.}~\bibnamefont{Brunner}}, \bibinfo
  {author} {\bibfnamefont{D.}~\bibnamefont{Cavalcanti}}, \bibinfo {author}
  {\bibfnamefont{S.}~\bibnamefont{Pironio}}, \bibinfo {author}
  {\bibfnamefont{V.}~\bibnamefont{Scarani}},\ and\ \bibinfo {author}
  {\bibfnamefont{S.}~\bibnamefont{Wehner}},\ }%
  \bibfield{journal}{%
  \Doi{10.1103/RevModPhys.86.419}{\bibinfo {journal} {Rev. Mod. Phys.}}\ }%
  \textbf{\bibinfo {volume} {86}},\ \bibinfo {pages} {419} (\bibinfo {month}
  {Apr}\ \bibinfo {year} {2014}),\
  \url{http://link.aps.org/doi/10.1103/RevModPhys.86.419}%
  \bibAnnoteFile{NoStop}{Cavalcanti14}%
\bibitem{Optomechanics09}%
  \BibitemOpen
  \bibfield{author}{%
  \bibinfo {author} {\bibfnamefont{F.}~\bibnamefont{Marquardt}}\ and\ \bibinfo
  {author} {\bibfnamefont{S.~M.}\ \bibnamefont{Girvin}},\ }%
  \bibfield{journal}{%
  \Doi{10.1103/Physics.2.40}{\bibinfo {journal} {Physics}}\ }%
  \textbf{\bibinfo {volume} {2}},\ \bibinfo {pages} {40} (\bibinfo {month}
  {May}\ \bibinfo {year} {2009}),\
  \url{http://link.aps.org/doi/10.1103/Physics.2.40}%
  \bibAnnoteFile{NoStop}{Optomechanics09}%
\bibitem{Optomechanics10}%
  \BibitemOpen
  \bibfield{author}{%
  \bibinfo {author} {\bibfnamefont{M.}~\bibnamefont{Aspelmeyer}}, \bibinfo
  {author} {\bibfnamefont{S.}~\bibnamefont{Gr\"{o}blacher}}, \bibinfo {author}
  {\bibfnamefont{K.}~\bibnamefont{Hammerer}},\ and\ \bibinfo {author}
  {\bibfnamefont{N.}~\bibnamefont{Kiesel}},\ }%
  \bibfield{journal}{%
  \Doi{10.1364/JOSAB.27.00A189}{\bibinfo {journal} {J. Opt. Soc. Am. B}}\ }%
  \textbf{\bibinfo {volume} {27}},\ \bibinfo {pages} {A189} (\bibinfo {month}
  {Jun}\ \bibinfo {year} {2010}),\
  \url{http://josab.osa.org/abstract.cfm?URI=josab-27-6-A189}%
  \bibAnnoteFile{NoStop}{Optomechanics10}%
\bibitem{Anders10}%
  \BibitemOpen
  \bibfield{author}{%
  \bibinfo {author} {\bibfnamefont{K.}~\bibnamefont{Hammerer}}, \bibinfo
  {author} {\bibfnamefont{A.~S.}\ \bibnamefont{S\o{}rensen}},\ and\ \bibinfo
  {author} {\bibfnamefont{E.~S.}\ \bibnamefont{Polzik}},\ }%
  \bibfield{journal}{%
  \Doi{10.1103/RevModPhys.82.1041}{\bibinfo {journal} {Rev. Mod. Phys.}}\ }%
  \textbf{\bibinfo {volume} {82}},\ \bibinfo {pages} {1041} (\bibinfo {month}
  {Apr}\ \bibinfo {year} {2010}),\
  \url{http://link.aps.org/doi/10.1103/RevModPhys.82.1041}%
  \bibAnnoteFile{NoStop}{Anders10}%
\bibitem{Grangier03}%
  \BibitemOpen
  \bibfield{author}{%
  \bibinfo {author} {\bibfnamefont{J.}~\bibnamefont{Wenger}}, \bibinfo {author}
  {\bibfnamefont{M.}~\bibnamefont{Hafezi}}, \bibinfo {author}
  {\bibfnamefont{F.}~\bibnamefont{Grosshans}}, \bibinfo {author}
  {\bibfnamefont{R.}~\bibnamefont{Tualle-Brouri}},\ and\ \bibinfo {author}
  {\bibfnamefont{P.}~\bibnamefont{Grangier}},\ }%
  \bibfield{journal}{%
  \Doi{10.1103/PhysRevA.67.012105}{\bibinfo {journal} {Phys. Rev. A}}\ }%
  \textbf{\bibinfo {volume} {67}},\ \bibinfo {pages} {012105} (\bibinfo {month}
  {Jan}\ \bibinfo {year} {2003}),\
  \url{http://link.aps.org/doi/10.1103/PhysRevA.67.012105}%
  \bibAnnoteFile{NoStop}{Grangier03}%
\bibitem{Glaubook}%
  \BibitemOpen
  \bibfield{author}{%
  \bibinfo {author} {\bibfnamefont{J.}~\bibnamefont{Glauber}},\ }%
  \emph{\bibinfo {title} {Quantum Theory of Optical Coherence}}\ (\bibinfo
  {publisher} {WILEY-VCH Verlag GmbH and Co. KGaA},\ \bibinfo {year} {2007})%
  \bibAnnoteFile{NoStop}{Glaubook}%
\bibitem{Vogel99}%
  \BibitemOpen
  \bibfield{author}{%
  \bibinfo {author} {\bibfnamefont{W.~V.}\ \bibnamefont{D.-G.~Welsch}}\ and\
  \bibinfo {author} {\bibfnamefont{T.}~\bibnamefont{Opatrny}},\ }%
  \bibfield{journal}{%
  \bibinfo {journal} {Progress in Optics}\ }%
  \textbf{\bibinfo {volume} {39}},\ \bibinfo {pages} {63} (\bibinfo {year}
  {1999})%
  \bibAnnoteFile{NoStop}{Vogel99}%
\bibitem{Vogel00}%
  \BibitemOpen
  \bibfield{author}{%
  \bibinfo {author} {\bibfnamefont{W.}~\bibnamefont{Vogel}},\ }%
  \bibfield{journal}{%
  \Doi{10.1103/PhysRevLett.84.1849}{\bibinfo {journal} {Phys. Rev. Lett.}}\ }%
  \textbf{\bibinfo {volume} {84}},\ \bibinfo {pages} {1849} (\bibinfo {month}
  {Feb}\ \bibinfo {year} {2000}),\
  \url{http://link.aps.org/doi/10.1103/PhysRevLett.84.1849}%
  \bibAnnoteFile{NoStop}{Vogel00}%
\bibitem{shapiro}%
  \BibitemOpen
  \bibfield{author}{%
  \bibinfo {author} {\bibfnamefont{A.~I.}\ \bibnamefont{Lvovsky}}\ and\
  \bibinfo {author} {\bibfnamefont{J.~H.}\ \bibnamefont{Shapiro}},\ }%
  \bibfield{journal}{%
  \Doi{10.1103/PhysRevA.65.033830}{\bibinfo {journal} {Phys. Rev. A}}\ }%
  \textbf{\bibinfo {volume} {65}},\ \bibinfo {pages} {033830} (\bibinfo {month}
  {Feb}\ \bibinfo {year} {2002}),\
  \url{http://link.aps.org/doi/10.1103/PhysRevA.65.033830}%
  \bibAnnoteFile{NoStop}{shapiro}%
\bibitem{Vogel04}%
  \BibitemOpen
  \bibfield{author}{%
  \bibinfo {author} {\bibfnamefont{E.}~\bibnamefont{Shchukin}}, \bibinfo
  {author} {\bibfnamefont{T.}~\bibnamefont{Richter}},\ and\ \bibinfo {author}
  {\bibfnamefont{W.}~\bibnamefont{Vogel}},\ }%
  \bibfield{journal}{%
  \bibinfo {journal} {Journal of Optics B: Quantum and Semiclassical Optics}\
  }%
  \textbf{\bibinfo {volume} {6}},\ \bibinfo {pages} {S597} (\bibinfo {year}
  {2004}),\ \url{http://stacks.iop.org/1464-4266/6/i=6/a=020}%
  \bibAnnoteFile{NoStop}{Vogel04}%
\bibitem{Belzig11}%
  \BibitemOpen
  \bibfield{author}{%
  \bibinfo {author} {\bibfnamefont{A.}~\bibnamefont{Bednorz}}\ and\ \bibinfo
  {author} {\bibfnamefont{W.}~\bibnamefont{Belzig}},\ }%
  \bibfield{journal}{%
  \Doi{10.1103/PhysRevA.83.052113}{\bibinfo {journal} {Phys. Rev. A}}\ }%
  \textbf{\bibinfo {volume} {83}},\ \bibinfo {pages} {052113} (\bibinfo {month}
  {May}\ \bibinfo {year} {2011}),\
  \url{http://link.aps.org/doi/10.1103/PhysRevA.83.052113}%
  \bibAnnoteFile{NoStop}{Belzig11}%
\bibitem{17Hilbert}%
  \BibitemOpen
  \bibfield{author}{%
  \bibinfo {author} {\bibfnamefont{J.}~\bibnamefont{Korbicz}}, \bibinfo
  {author} {\bibfnamefont{J.}~\bibnamefont{Cirac}}, \bibinfo {author}
  {\bibfnamefont{J.}~\bibnamefont{Wehr}},\ and\ \bibinfo {author}
  {\bibfnamefont{M.}~\bibnamefont{Lewenstein}},\ }%
  \bibfield{journal}{%
  \Doi{10.1103/PhysRevLett.94.153601}{\bibinfo {journal} {Phys. Rev. Lett.}}\
  }%
  \textbf{\bibinfo {volume} {94}},\ \bibinfo {pages} {153601} (\bibinfo {month}
  {Apr}\ \bibinfo {year} {2005}),\
  \url{http://link.aps.org/doi/10.1103/PhysRevLett.94.153601}%
  \bibAnnoteFile{NoStop}{17Hilbert}%
\bibitem{Deans}%
  \BibitemOpen
  \bibfield{author}{%
  \bibinfo {author} {\bibfnamefont{S.}~\bibnamefont{Deans}},\ }%
  \emph{\bibinfo {title} {The Radon Transform and some Its Applications}}\
  (\bibinfo {publisher} {A Wiley-Interscience Publication},\ \bibinfo {year}
  {1983})%
  \bibAnnoteFile{NoStop}{Deans}%
\bibitem{Markoe}%
  \BibitemOpen
  \bibfield{author}{%
  \bibinfo {author} {\bibfnamefont{A.}~\bibnamefont{Markoe}},\ }%
  \emph{\bibinfo {title} {Analytic Tomography, Encyclopedia of Mathematics and
  its Applications}}\ (\bibinfo {publisher} {Cambridge Uni. Press},\ \bibinfo
  {year} {2006})%
  \bibAnnoteFile{NoStop}{Markoe}%
\bibitem{Niev86}%
  \BibitemOpen
  \bibfield{author}{%
  \bibinfo {author} {\bibfnamefont{Y.}~\bibnamefont{Nievergelt}},\ }%
  \bibfield{journal}{%
  \bibinfo {journal} {Journal of Math. Anal. and Appl.}\ }%
  \textbf{\bibinfo {volume} {120}},\ \bibinfo {pages} {288} (\bibinfo {year}
  {1986})%
  \bibAnnoteFile{NoStop}{Niev86}%
\bibitem{Optomechanics14}%
  \BibitemOpen
  \bibfield{author}{%
  \bibinfo {author} {\bibfnamefont{M.}~\bibnamefont{Aspelmeyer}}, \bibinfo
  {author} {\bibfnamefont{T.~J.}\ \bibnamefont{Kippenberg}},\ and\ \bibinfo
  {author} {\bibfnamefont{F.}~\bibnamefont{Marquardt}},\ }%
  \bibfield{journal}{%
  \Doi{10.1103/RevModPhys.86.1391}{\bibinfo {journal} {Rev. Mod. Phys.}}\ }%
  \textbf{\bibinfo {volume} {86}},\ \bibinfo {pages} {1391} (\bibinfo {month}
  {Dec}\ \bibinfo {year} {2014}),\
  \url{http://link.aps.org/doi/10.1103/RevModPhys.86.1391}%
  \bibAnnoteFile{NoStop}{Optomechanics14}%
\end{thebibliography}

%\begin{thebibliography}{9}

%

\end{document}